%% file: relational-connectors.tex
\documentclass[runningheads,final,UKenglish,envcountsame,envcountsect]{llncs}
 \usepackage{microtype}

\usepackage{enumitem}
\usepackage{xspace}
\usepackage{xcolor}
\usepackage{soul}
\usepackage[utf8]{inputenc}
\usepackage{amsmath,amssymb,mathrsfs}
\usepackage{stmaryrd} 
\usepackage{wasysym}

\usepackage{etex,etoolbox}

\newcommand{\partialto}{\rightharpoonup}
\newcommand{\partialtone}{\partialto_{\mathrm{ne}}}
\newcommand{\dom}{\mathrm{dom}}
\newcommand{\into}{\hookrightarrow}

\newcommand{\Lor}{\bigvee}
\newcommand{\Land}{\bigwedge}
\newcommand{\doublemod}[1]{\langle #1\rangle}
\newcommand{\predlift}[1]{\mathsf{PL}(#1)}
\newcommand{\frel}[2]{#1 \mathbin{\ooalign{$\rightarrow$\cr$\hspace{0.15ex}+$\cr}} #2}
\newcommand{\rev}[1]{#1 ^ \circ}
\newcommand{\idcon}[1]{\mathsf{Id}^\mathsf{c}_{#1}}
\newcommand{\natcomp}{\bullet}
\newcommand{\posbool}{\mathsf{pos}}
\newcommand{\Pos}{\mathsf{Pos}}
\newcommand{\logcomp}{\mathbin{\blacktriangleright}}
\newcommand{\forthcon}{L_{\mathsf{f}}}
\newcommand{\tracecon}{L_{\mathsf{t}}}
\newcommand{\simul}{\preceq}
\newcommand{\bisim}{\simeq}
\newcommand{\by}[1]{\text{(#1)}}

\newcommand{\resetCurThmBraces}{%
\gdef\curThmBraceOpen{(}%
\gdef\curThmBraceClose{)}}
\resetCurThmBraces
\newcommand{\removeThmBraces}{%
\gdef\curThmBraceOpen{}%
\gdef\curThmBraceClose{}}
\resetCurThmBraces

\usepackage{etoolbox}
\patchcmd{\thmhead}{(#3)}{\curThmBraceOpen #3\curThmBraceClose}{}{}

\makeatletter
\RequirePackage[bookmarks,unicode,colorlinks=true]{hyperref}%
   \def\@citecolor{blue}%
   \def\@urlcolor{blue}%
   \def\@linkcolor{blue}%

\def\orcidID#1{\smash{\href{http://orcid.org/#1}{\protect\raisebox{-1.25pt}{\protect\includegraphics{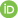}}}}}
\makeatother

\newcommand{\id}{\mathit{id}}

\numberwithin{equation}{section}

\usepackage{algorithmicx}
    \usepackage[ruled]{algorithm}
\usepackage[noend]{algpseudocode}

\usepackage{tikz}
 \usetikzlibrary{trees}
 \usetikzlibrary{shapes}
 \usetikzlibrary{fit}
 \usetikzlibrary{shadows}
 \usetikzlibrary{backgrounds}
 \usetikzlibrary{arrows,automata}

\tikzset{
   n/.style= {circle,fill,inner sep=1.5pt,node distance=2cm}
  ,acc/.style={circle,draw,inner sep=3pt,node distance=2cm}
  ,phantom/.style={circle},
  ,arr/.style={->, >=stealth, semithick, shorten <= 3pt, shorten >= 3pt}
}

\renewcommand{\Box}{\square}
\renewcommand{\Diamond}{\lozenge}
\newcommand{\pow}{\mathcal{P}}

\makeatletter
\def\moverlay{\mathpalette\mov@rlay}
\def\mov@rlay#1#2{\leavevmode\vtop{%
   \baselineskip\z@skip \lineskiplimit-\maxdimen
   \ialign{\hfil$\m@th#1##$\hfil\cr#2\crcr}}}
\newcommand{\charfusion}[3][\mathord]{
    #1{\ifx#1\mathop\vphantom{#2}\fi
        \mathpalette\mov@rlay{#2\cr#3}
      }
    \ifx#1\mathop\expandafter\displaylimits\fi}
\makeatother

\newcommand{\FA}{{\mathfrak A}}

\newcommand{\Dist}{\mathcal{D}}

\newcommand{\Nat}{{\mathbb{N}}}

\newcommand{\dual}[1]{\overline{#1}}
\newcommand{\FLang}{\mathcal{F}}
\newcommand{\Lang}{\mathcal{L}}

\newcommand{\Set}{\mathsf{Set}}

\newcommand{\Sem}[1]{{[\![#1]\!]}}

\newcommand{\A}{\mathcal{A}}
\newcommand{\B}{\mathcal{B}}
\newcommand{\C}{\mathcal{C}}

\newcommand{\takeout}[1]{\empty}

\makeatletter
\newcommand{\mysubsec}[1]{%
  \par\medskip\noindent{\bfseries\sffamily #1}\hspace{3mm}%
  \@ifnextchar\par{\@gobble}{}
}
\makeatother

\input{defslncs}

\newcommand{\BREL}{\mathsf{BinRel}}
\newcommand{\ORD}{\mathsf{PreOrd}}
\newcommand{\catA}{\mathsf{A}}
\newcommand{\catB}{\mathsf{B}}

\begin{document}
\title{Relational Connectors and \\Heterogeneous Simulations}
\titlerunning{Relational Connectors and Heterogeneous Bisimulations}
\author{Pedro Nora\inst{2}\orcidID{0000-0001-8581-0675}\and Jurriaan Rot\inst{2}\orcidID{0000-0002-1404-6232}\thanks{Funded by the 
Dutch Research Council (NWO) -- project number~VI.Vidi.223.096} \and Lutz Schröder\inst{1}\orcidID{0000-0002-3146-5906}\thanks{Funded by the Deutsche Forschungsgemeinschaft (DFG, German Research Foundation) -- project number 531706730}\and Paul Wild\inst{1}\orcidID{0000-0001-9796-9675}}
\institute{Friedrich-Alexander-Universität Erlangen-Nürnberg \and Radboud University, Nijmegen}
\maketitle

\begin{abstract}
  While behavioural equivalences among systems of the same type, such as Park/Milner bisimilarity of labelled transition systems, are an established notion, a systematic treatment of relationships between systems of different types is currently missing. We provide such a treatment in the framework of universal coalgebra, in which the type of a system (nondeterministic, probabilistic, weighted, game-based etc.) is abstracted as a set functor: We introduce \emph{relational   connectors} among set functors, which induce notions of heterogeneous (bi)simulation among coalgebras of the respective types. We give a number of constructions on relational connectors. In particular, we identify composition and converse operations on relational connectors; we construct corresponding identity relational connectors, showing that the latter generalize the standard Barr extension of weak-pullback-preserving functors; and we introduce a Kantorovich construction in which relational connectors are induced from relations between modalities. For Kantorovich relational connectors, one has a notion of dual-purpose modal logic interpreted over both system types, and we prove a corresponding Hennessy-Milner-type theorem stating that generalized (bi)similarity coincides with theory inclusion on finitely-branching systems. We apply these results to a number of example scenarios involving labelled transition systems with different label alphabets, probabilistic systems, and input/output conformances.
\end{abstract}

\section{Introduction}\label{sec:intro}

Notions of simulation and bisimulation are pervasive in the
specification and verification of reactive systems
(e.g.~\cite{milner89}). For instance, they appear in state space
reduction (e.g.~\cite{BlomOrzan05}), they are used to specify concrete
systems in terms of abstract systems (e.g.\ in connection with the
analysis of ePassport protocols~\cite{HorneMauw21}), and, classically,
they relate tightly to indistinguishability in modal
logic~\cite{HennessyMilner85}. Originally introduced for (labelled)
transition systems, notions of (bi)simulation have been extended to a
wide range of system types, e.g.\ probabilistic
systems~\cite{LarsenSkou91,DesharnaisEA02}, weighted
systems~\cite{Buchholz08}, or monotone neighbourhood
frames~\cite{Pauly99,HansenKupke04}. They have received a uniform
treatment in the framework of universal
coalgebra~\cite{Rutten00}. However, so far, notions of (bi)simulation
have typically been confined to settings where the two systems being
compared are of the same type in a strict sense, e.g.\ labelled
transition systems (LTS) over the same alphabet. In the present paper,
we introduce a principled approach to comparing behaviour across
different system types by means of \emph{heterogeneous
  (bi)simulations}.

To this end, we encapsulate system types as set functors in the
paradigm of universal coalgebra, and introduce \emph{(relational)
  connectors} between system types. The latter generalize \emph{lax
  extensions}, which induce notions of (bi)simulation on a single
system type~\cite{MartiVenema12,MartiVenema15}. A connector between
functors~$F$ and~$G$ induces a notion of (bi)simulation between
$F$-coalgebras and $G$-coalgebras, i.e.\ between the systems of the
types represented by~$F$ and~$G$, respectively, for instance between
nondeterministic and probabilistic systems. We give a range of
constructions of connectors, such as converse, composition, and
pulling back along natural transformations. Notably, we show that the
composition of relational connectors admits identities. Identity
relational connectors satisfy a minimality condition, and form
smallest lax extensions of functors; for weak-pullback-preserving
functors, they coincide with the \emph{Barr extension}~\cite{Barr93},
which instantiates, e.g., to the well-known Egli-Milner relation
lifting for the powerset functor. We use these constructions to cover
a number of application scenarios, e.g.\ transferring bisimilarity
among LTS over different alphabets; sharing of infinite traces among
LTS; nondeterministic abstractions of probabilistic LTS; and
input-output conformances (\emph{ioco})~\cite{BosJM19}.

We go on to give a construction of relational connectors based on
relating modalities, modelled as predicate liftings in the style of
coalgebraic logic~\cite{Pattinson04,Schroder08}. In reference to
constructions of behavioural metrics (on a single system type) from
modalities~\cite{BaldanEA18,WildSchroder22}, we call such relational
connectors \emph{Kantorovich}. Many of our running examples turn out
to be Kantorovich. We then prove a Hennessy-Milner-type result for
Kantorovich connectors, showing that on finitely branching systems,
the induced similarity coincides with theory inclusion in a generic
dual-purpose modal logic that can be interpreted over both of the
involved system types. The generic theorem instantiates to logical
characterizations of bisimulation between LTS with different
alphabets, trace sharing between LTS,  nondeterministic abstraction
of probabilistic LTS, and ioco compatibility.

Proofs are often omitted or only sketched; details can be found in the
full version~\cite{NoraEA24arxiv}.

\paragraph*{Related work} Relational connectors
generalize lax extensions~\cite{Sea05,SS08,MartiVenema12,MartiVenema15}, which
belong to an extended strand of work on extending set functors to act
on
relations~(e.g.~\cite{BackhouseEA91,ThijsThesis,HT00,Lev11}). The
Kantorovich construction similarly generalizes constructions of
functor liftings and lax extensions in both two-valued and
quantitative
settings~\cite{MartiVenema12,MartiVenema15,GorinSchroder13,BaldanEA18,WildSchroder22,GoncharovEA23}.
Our heterogeneous Hennessy-Milner theorem generalizes (the
monotone case of) coalgebraic Hennessy-Milner theorems for behavioural
equivalence~\cite{Pattinson04,Schroder08} and behavioural
preorders~\cite{KapulkinEA12,WildSchroder21}.
A different generalization of notions of bisimulation occurs via
functor liftings along fibrations~\cite{DBLP:journals/iandc/HermidaJ98,DBLP:journals/mscs/HasuoKC18},
which have also been connected to modal logics~\cite{DBLP:journals/lmcs/KupkeR21,KomoridaEA21}.
The Kantorovich construction is generalized
there by the so-called codensity lifting~\cite{DBLP:journals/logcom/SprungerKDH21}. Heterogeneous
notions of bisimulation have not been considered there.

\section{Preliminaries: Coalgebras and Lax Extensions}
\label{sec:prelims}
We assume basic familiarity with category theory
(e.g.~\cite{AdamekEA90}). We proceed to recall requisite background on
relations, coalgebras, and lax extensions.

\subsubsection*{Relations} A \emph{relation} from a set~$X$ to a
set~$Y$ is a subset $r\subseteq X\times Y$, denoted
$r\colon\frel{X}{Y}$; we write $x\mathrel{r}y$ for $(x,y)\in r$. Given
$r\colon\frel{X}{Y}$ and $s\colon\frel{Y}{Z}$, we write $s\cdot r$ for
the applicative-order relational composite of~$r$ and~$s$, i.e.
\begin{equation*}
  s\cdot r=\{(x,z)\mid\exists y\in Y.\,x\mathrel{r}y\mathrel{s}z\}.
\end{equation*}
The \emph{join} of a family of relations is just its union. Relational
composition is \emph{join continuous} in both arguments, i.e.\ we have
$(\bigvee_{i\in I}s_i)\cdot r=\bigvee_{i\in I}(s_i\cdot r)$ and
$s\cdot (\bigvee_{i\in I}r_i)=\bigvee_{i\in I}(s\cdot r_i)$.  We
define the \emph{relational converse} $\rev{r}\colon\frel{Y}{X}$ by
\begin{math}
  \rev{r}=\{(y,x)\mid (x,y)\in r\}.
\end{math}
We identify a function $f\colon X\to Y$ with its graph, i.e.\ the
relation $\{(x,f(x))\mid x\in X\}$. For clarity, we sometimes write
\begin{math}
  \Delta_X=\{(x,x)\mid x\in X\}
\end{math}
for the diagonal relation on~$X$, i.e.\ the graph of the identity
function on~$X$, which is neutral for relational
composition. Functions $f\colon X\to Y$ are characterized by the
inequalities
\begin{equation*}
  \Delta_X\subseteq\rev{f}\cdot f\quad\text{\emph{(totality)}} \qquad
  f\cdot\rev{f}\subseteq\Delta_Y \quad \text{\emph{(univalence)}.}
\end{equation*}
Given a subset $A\subseteq X$ and a relation $r\colon\frel{X}{Y}$, we
write $r[A]=\{y\in Y\mid\exists x\in A.\,x\mathrel{r}y\}$ for the
\emph{relational image} of~$A$ under~$r$. We say that~$r$ is
\emph{right total} if $r[X]=Y$, and \emph{left total} if
$\rev{r}[Y]=X$.

\subsubsection*{Universal coalgebra} State-based systems of a wide
range of transition types can be usefully abstracted as coalgebras for
a given functor encapsulating the system type~\cite{Rutten00}. We work
more specifically over the category of sets, and thus model a system
type as a functor $F\colon\Set\to\Set$. Then, an \emph{$F$-coalgebra}
is a pair $(C,\gamma)$ consisting of a set~$C$ of \emph{states} and a
\emph{transition map} $\gamma\colon C\to FC$. Following tradition in
algebra, we often just write~$C$ for the coalgebra $(C,\gamma)$. We
think of~$C$ as a set of \emph{states}, and of~$\gamma$ as assigning
to each state~$c\in C$ a collection~$\gamma(c)$ of successor states,
structured according to~$F$. For instance, if~$F=\pow$ is the usual
(covariant) powerset functor, then~$\gamma$ assigns to each state a
\emph{set} of successors, so a $\pow$-coalgebra is just a standard
relational transition system. More generally, given a set~$\A$ of
\emph{labels}, $F$-coalgebras for the functor $F=\pow(\A\times(-))$
are $\A$-labelled transition systems ($\A$-LTS). On the other hand, we
write $\Dist$ for the (discrete) \emph{distribution functor}, which
assigns to a set~$X$ the set of discrete probability distributions
on~$X$ (which may be represented as functions $\alpha\colon X\to[0,1]$
such that $\sum_{x\in X}\alpha(x)=1$, extended to subsets
$A\subseteq X$ by $\alpha(A)=\sum_{x\in A}\alpha(x)$) and acts on maps
by taking direct images. Then, $\Dist$-coalgebras are probabilistic
transition systems, or Markov chains, while
$\Dist(\A\times(-))$-coalgebras are probabilistic $\A$-labelled
transition systems (probabilistic $\A$-LTS). We assume w.l.o.g.\ that
functors preserve injective maps~\cite{Barr93}, and then in fact that
subset inclusions are preserved.

A \emph{morphism} $f\colon C\to D$ of $F$-coalgebras $(C,\gamma)$,
$(D,\delta)$ is a map $f\colon C\to D$ such that
$Ff\cdot\gamma=\delta\cdot f$. States $c\in C$, $d\in D$ in
$F$-coalgebras $C,D$ are \emph{behaviourally equivalent} if there
exist an $F$-coalgebra~$(E,\epsilon)$ and morphisms $f\colon C\to E$,
$g\colon D\to E$ such that $f(x)=g(y)$. For instance, morphisms of
$\pow(\A\times(-))$-coalgebras are bounded morphisms of $\A$-LTS in
the usual sense (i.e.\ functional bisimulations), and behavioural
equivalence instantiates to the usual notion of (strong) bisimilarity
on LTS.

\subsubsection*{Lax extensions} As indicated in the introduction,
relational connectors are largely intended as a generalization of lax
extensions, which extend a single functor to act also on relations,
to settings where relations need to connect elements of different
functors. A \emph{lax extension}~$L$ (references are in
\autoref{sec:intro}) of a set functor~$F$ assigns to each relation
$r\colon\frel{X}{Y}$ a relation $Lr\colon\frel{FX}{FY}$ such that
\begin{align*}
  \text{(L1)} &\quad r_1 \subseteq r_2 \rightarrow Lr_1 \subseteq Lr_2 &&\text{\emph{(monotonicity)}} \\
  \text{(L2)} &\quad  Ls\cdot Lr \subseteq L(s\cdot r) &&\text{\emph{(lax functoriality)}}\\
  \text{(L3)} &\quad  Ff\subseteq Lf \text{ and } \rev{(Ff)}\subseteq L(\rev{f})
\end{align*}
for all sets $X,Y,Z$, and $r,r_1,r_2\colon\frel{X}{Y}$,
$s\colon\frel{Y}{Z}$, $f\colon X\to Y$.  These conditions imply
\emph{naturality}~\cite{Sea05,MartiVenema15}:
\begin{equation*}
  L(\rev{g}\cdot r\cdot f)=\rev{(Fg)}\cdot Lr\cdot Ff
\end{equation*}
for $r\colon\frel{X}{Y}$ and maps $f\colon X'\to X$,
$g\colon Y'\to Y$.  We say that $L$ \emph{preserves diagonals} if
\begin{math}
  L\Delta_X \subseteq \Delta_{FX}\text{ for all~$X$},
\end{math}
equivalently, 
  $L f \subseteq F f$ for all maps~$f$. 
Moreover,~$L$ \emph{preserves converse} if $L(\rev{r})=\rev{(Lr)}$ for
all~$r$. (Indeed, this property is often included in the definition of
lax extension~\cite{MartiVenema15}.)

Lax extensions induce notions of \emph{(bi)simulation}, that is, of
relations that witness behavioural equivalence in the sense recalled
above. Given a lax extension~$L$ of a functor~$F$, a relation
$r\colon\frel{C}{D}$ between $F$-coalgebras~$(C,\gamma)$, $(D,\delta)$
is an \emph{$L$-simulation} if $\delta\cdot r\le Lr\cdot\gamma$; that
is, whenever $c\mathrel{r}d$, then
$\gamma(c)\mathrel{Lr}\delta(d)$. Two states $c\in C$, $d\in D$ are
\emph{$L$-similar} if there exists an $L$-simulation
$r\colon\frel{C}{D}$ such that $c\mathrel{r}d$. If~$L$ preserves
converse, then the converse $\rev{r}$ of an $L$-simulation~$r$ is also
an~$L$-simulation and, hence, $L$-similarity is symmetric; one thus
speaks more appropriately of \emph{$L$-bisimulations} and
\emph{$L$-bisimilarity}. Notably, if~$L$ preserves converse and
diagonals, then $L$-bisimilarity coincides with behavioural
equivalence~\cite{MartiVenema12,MartiVenema15}. Every lax extension
can be induced from a choice of
modalities~\cite{MartiVenema12,MartiVenema15}; we return to this point
in \autoref{sec:kantorovich}. We recall only the most basic example:
\begin{example}\label{expl:LTS-lax}
  Let $\A$ be a set of labels, and let $F=\pow(\A\times(-))$ be the
  functor modelling $\A$-LTS as recalled above. We have a converse-
  and diagonal-preserving lax extension~$L$ of~$F$ given by
  $S\mathrel{Lr}T$ iff (i) for all $(l,x)\in S$, there is $(l,y)\in T$
  such that $x\mathrel{r}y$ (`\emph{forth}'), and (ii) for all
  $(l,y)\in T$, there is $(l,x)\in S$ such that $x\mathrel{r}y$
  (`\emph{back}'). Indeed,~$L$ is even a strict extension, i.e.\
  condition~(L2) holds in the stronger form $Ls\cdot Lr=L(s\cdot r)$
  for composable $s,r$ (such strict extensions exist, and then are
  unique, iff the underlying functor preserves weak
  pullbacks~\cite{Barr70,Trnkova80}).  $L$-bisimulations in the sense
  recalled above are precisely (strong) bisimulations of LTS in the
  standard sense.
\end{example}
\begin{remark}[Barr extension] \label{rem:barr} The above-mentioned
  strict extension~$L$ of a weak-pullback-preserving functor~$F$,
  often called the \emph{Barr extension}, is described as
  follows~\cite{Barr70}: A relation~$r\colon\frel{X}{Y}$ itself forms
  a set (a subset of $X\times Y$), and as such comes with two
  projection maps $\pi_1\colon r\to X$, $\pi_2\colon r\to Y$. Then,
  $Lr=F\pi_2\cdot\rev{(F\pi_1)}$. A slightly simpler example than
  \autoref{expl:LTS-lax} is the Barr extension~$L$ of the powerset
  functor~$\pow$, which coincides with the well-known Egli-Milner
  extension: For $r\colon\frel{X}{Y}$ and $S\in\pow(X)$,
  $T\in\pow(Y)$, we have $S\mathrel{Lr}T$ iff for every $x\in S$ there
  is $y\in T$ such that $x\mathrel{r}y$ and symmetrically.
\end{remark}

\section{Relational Connectors}
\label{sec:rel-conn}
\noindent We proceed to introduce relational connectors and associated
constructions.

\subsection{Axiomatics}
The main idea is that while a lax extension of
a functor~$F$ (\autoref{sec:prelims}) lifts relations between sets~$X$
and~$Y$ to relations between~$FX$ and~$FY$, a relational connector
between functors~$F$ and~$G$ lifts relations between sets~$X$ and~$Y$
to relations between~$FX$ and~$GY$. The axiomatics of relational
connectors is inspired by that of lax extensions, but forcibly
deviates in some respects:

\begin{defn}[Relational connector]\label{def:relational-connector}
  Let $F,G$ be set functors. A \emph{relational connector} (or
  occasionally just a \emph{connector}) $L\colon F\to G$ assigns to
  each relation $r\colon\frel{X}{Y}$ a relation
  \begin{equation*}
    Lr\colon \frel{FX}{GY}
  \end{equation*}
  such that the following conditions hold:
  \begin{enumerate}
  \item Whenever $r_1\subseteq r_2$ for $r_1,r_2\colon\frel{X}{Y}$,
    then $Lr_1\subseteq Lr_2$ (\emph{monotonicity}).
  \item Whenever $f\colon X'\to X$, $g\colon Y'\to Y$, and
    $r\colon\frel{X}{Y}$, then
    \begin{equation*}
      L(\rev{g}\cdot r\cdot f)=\rev{(Gg)}\cdot Lr\cdot Ff\qquad\text{(\emph{naturality}).}
    \end{equation*}
  \end{enumerate}
  We define an ordering on connectors $F\to G$ by $L\le K$ iff
  $Lr\subseteq Kr$ for all~$r$.
\end{defn}
In pointful notation, naturality says that for data as above and
$a\in FX'$, $b\in GY'$, we have
\begin{equation}\label{eq:nat-pointwise}
  Ff(a)\mathrel{Lr}Gg(b)\quad\text{iff}\quad a\mathrel{L(\rev{g}\cdot r\cdot f)}b.
\end{equation}
\begin{example}\label{expl:lts-connector}
  Let $F=\pow(\A\times(-))$, $G=\pow(\B\times(-))$ be the functors
  determining $\A$-LTS and $\B$-LTS as their coalgebras, respectively
  (\autoref{sec:prelims}). For $R\subseteq \A\times\B$, we define a
  relational connector~$L_R\colon F\to G$ by
  \begin{align*}
    S\mathrel{L_Rr}T \iff \forall (l,m)\in R.
    & \;\forall (l,x)\in S.\,\exists (m,y)\in T.\,x\mathrel{r}y\;\land\\
    & \;\forall (m,y)\in T.\,\exists (l,x)\in S.\,x\mathrel{r}y
  \end{align*}
  for $r\colon \frel{X}{Y}$. We will later use instances of this type
  of relational connector to transfer bisimilarity between $\A$-LTS
  and $\B$-LTS.
\end{example}
\noindent Of course, every lax extension of~$F$ is a relational
connector $F\to F$. In the axiomatics of relational connectors,
notable omissions in comparison to lax extensions include~(L2)
and~(L3), both of which in general just fail to type for relational
connectors. We will later discuss these conditions and further ones as
properties that a relational connector may or may not have, if
applicable. Note that we do retain an important consequence of these
properties, viz., naturality.

\subsection{Constructions}
Our perspective on relational connectors is partly driven by
constructions enabled by the axiomatics; maybe the most central ones
among these are composition and identities, introduced next.

\begin{defn}[Composition of relational connectors]
Given relational connectors $K\colon F\to G$, $L\colon G\to H$, we
define the \emph{composite} $L\cdot K\colon F\to H$ by
\begin{equation}\label{eq:comp}
  (L\cdot K)r=\bigvee_{r=s\cdot t}Ls\cdot Kt\quad\text{for $r\colon \frel{X}{Z}$},
\end{equation}
where the join is over all $t\colon \frel{X}{Y}$, $s\colon\frel{Y}{Z}$
such that $s\cdot t=r$, with~$Y$ ranging over all sets (see however
\autoref{thm:comp-couniv} and~\autoref{lem:comp-subset}).
\end{defn}
\begin{lemma}\label{lem:comp}
  Given relational connectors $K\colon F\to G$, $L\colon G\to H$, the
  composite $L\cdot K\colon F\to H$ is a relational connector.
\end{lemma}
\begin{proof}[sketch]
  \emph{Monotonicity:} Let $r\subseteq r'\colon \frel{X}{Z}$. If
  $a\;(L\cdot K)r\;c$ is witnessed by a factorization $r=s\cdot t$
  where $t\colon\frel{X}{Y}$, $s\colon\frel{Y}{Z}$, then
  $a\;(L\cdot K)r'\;c$ is witnessed by the factorization
  $r'=s'\cdot t'$ where $t'\colon \frel{X}{Y'}$,
  $s'\colon \frel{Y'}{Z}$ with $Y'=Y\cup (r'\setminus r)$ (w.l.o.g.\ a
  disjoint union) and
  \begin{equation*}
    t'=t\cup\{(x,(x,z))\mid (x,z)\in r'\setminus r\}\qquad
    s'=s\cup\{((x,z),z)\mid (x,z)\in r'\setminus r\}.
  \end{equation*}
  Remarkably, the further proof uses naturality (w.r.t.\ $Y\into Y'$)
  but not monotonicity of~$K$ and~$L$.

  \emph{Naturality:}
  $(L\cdot K)(\rev{g}\cdot r\cdot f)=\rev{(Hg)}\cdot(L\cdot K)r\cdot
  Ff$ is shown using naturality and monotonicity of~$K$ and $L$, monotonicity of
  $L\cdot K$, and totality and univalence of $f$ and $g$. \qed
\end{proof}
\noindent As an immediate consequence of monotonicity of composite
relational connectors, we have the following alternative description
of composition:
\begin{lemma}\label{lem:comp-subset}
  Given relational connectors $K\colon F\to G$, $L\colon G\to H$, we
  have
  \begin{equation*}\textstyle
    (L\cdot K)r=\bigvee_{r\supseteq s\cdot t}Ls\cdot Kt\quad\text{for $r\colon \frel{X}{Z}$}
  \end{equation*}
  where the join is over all $t\colon \frel{X}{Y}$,
  $s\colon\frel{Y}{Z}$ such that $r\supseteq s\cdot t$, with~$Y$
  ranging over all sets.
\end{lemma}
In order to compute composites of relational connectors, the following
observation is sometimes useful.
\begin{defn}\label{def:couniv-decomp}
  The \emph{couniversal factorization} $r=s\cdot t$ of a relation
  $r\colon\frel{X}{Z}$ is given by
  \begin{align*}
    Y & =\{(A,B)\in\pow(X)\times\pow(Z)\mid A\times B\subseteq r\}\\
    t & =\{(x,(A,B))\mid x\in A\}\colon\frel{X}{Y}\\
    s & =\{((A,B),z)\mid z\in B\}\colon\frel{Y}{Z}.
  \end{align*}
\end{defn}
\begin{lemma}\label{lem:couniv-decomp}
  Let $s\colon\frel{Y}{Z}$, $t\colon\frel{X}{Y}$ be the couniversal
  factorization of $r\colon\frel{X}{Z}$. Then indeed $r=s\cdot t$, and
  for every factorization $r=s'\cdot t'$ of~$r$ into
  $s'\colon\frel{Y'}{Z}$, $t'\colon\frel{X}{Y'}$, there is a map
  $f\colon Y'\to Y$ such that $s' = s \cdot f$ and
  $t' = \rev{f} \cdot t$.
\end{lemma}
\begin{theorem}\label{thm:comp-couniv}
  Let $K\colon F\to G$, $L\colon G\to H$ be relational connectors, and
  let $r=s\cdot t$ be the couniversal factorization of
  $r\colon\frel{X}{Z}$. Then
  \begin{equation*}
    (L\cdot K)r=Ls\cdot Kt.
  \end{equation*}
\end{theorem}
\noindent
We proceed to establish that the composition operation defined above
equips relational connectors with the structure of a quasicategory
(i.e.\ overlarge category). We first check associativity:
\begin{lemma}\label{lem:assoc}
  Let $K\colon F\to G$, $L\colon G\to H$, and $M\colon H\to V$ be
  relational connectors. Then $(M\cdot K)\cdot L=M\cdot (K\cdot L)$.
\end{lemma}
\noindent The straightforward proof uses join continuity of relational
composition. We next construct identities:
\begin{defn}[Identity relational connectors]
  The \emph{identity relational connector} $\idcon{F}\colon F\to F$
  on a set functor~$F$ is defined as follows. For $r\colon\frel{X}{Y}$,
  $b\in FX$, and $c\in FY$, we put $b\mathrel{\idcon{F}r} c$ iff for
  all set functors~$G$, all relational connectors~$L\colon G\to F$,
  all $s\colon\frel{Z}{X}$, and all $a\in GZ$,
  \begin{equation*}
    a\mathrel{Ls}b\quad\text{implies}\quad a\mathrel{L(r \cdot s)}c.
  \end{equation*}
\end{defn}
(This definition is highly impredicative, but we will later give a
characterization of~$\idcon{F}$ that eliminates quantification over
relational connectors.) We will show that~$\idcon{F}$ is neutral
w.r.t.\ composition of relational connectors. We first note that, as
an immediate consequence of the definition,
\begin{equation}
  \label{eq:idcon-diag}
  \Delta_{FX}\subseteq\idcon{F}\Delta_X\qquad\text{for all~$X$.}
\end{equation}
\begin{lemma}\label{lem:id}
  For each functor~$F$, $\idcon{F}$ is a relational connector.
\end{lemma}
The proof of naturality relies in particular on monotonicity of
relational connectors in combination with totality and univalence of
maps. We show next that identity connectors do actually act as
identities under composition:

\begin{lemma}\label{lem:id-comp}
  For each $L\colon G\to F$, we have
  $L=\idcon{F}\cdot L=L\cdot\idcon{G}$.
\end{lemma}
\begin{proof}[sketch]
  One shows, using~\eqref{eq:idcon-diag} inter alia, that~$\idcon{F}$
  is a left identity ($L=\idcon{F}\cdot L$). By a symmetric argument,
  composition of relational connectors also has right identities, and then the left
  and right identities are necessarily equal. \qed
\end{proof}
Relational connectors admit a natural notion of converse:
\begin{defn}[Converse, meet and product of relational connectors]
  The \emph{converse} $\rev{L}\colon G\to F$ of a relational connector
  $L\colon F\to G$ is given by
  \begin{equation*}
    \rev{L}r=\rev{(L\rev r)}\colon\frel{GX}{FY}
  \end{equation*}
  for $r\colon\frel{X}{Y}$. The \emph{meet} $L\cap K$ of relational
  connectors $L,K\colon F\to G$ is their componentwise intersection
  ($(L\cap K)r=Lr\cap Kr$).  For relational connectors
  $L_1 \colon F_1 \to G_1$ and $L_2 \colon F_2 \to G_2$, their
  \emph{product}
  $L_1 \times L_2 \colon F_1 \times F_2 \rightarrow G_1 \times G_2$ is
  given by
  \begin{equation*}
	  (a,b) \mathrel{(L_1 \times L_2) r} (c,d) \iff a \mathrel{L_1 r} c \text{ and } b \mathrel{L_2 r} d.
	\end{equation*}
\end{defn}
\begin{lemma}\label{lem:converse-meet-product}
	The converse, meet and product of relational connectors are again relational connectors.
\end{lemma}
\noindent We record some expected properties of converse:
\begin{lemma}\label{lem:rev-comp}
  Converse is involutive ($\rev{(\rev L)}=L$) and monotone. Moreover,
  for relational connectors $K\colon F\to G$ and $L\colon G\to H$, we
  have
    \begin{equation*}
      \rev{(L\cdot K)}=\rev K\cdot\rev L.
    \end{equation*}
\end{lemma}

\begin{remark}
  In view of the above properties, one may ask whether relational
  connectors form an overlarge allegory~\cite{FreydScedrov90}. We
  leave this question open for the moment; specifically, it is not
  clear that relational connectors satisfy the \emph{modular law}
  $(L\cdot K)\cap M \le L\cdot (K\cap (\rev{L}\cdot M))$.
\end{remark}

\begin{example}[Constructions of relational
  connectors]\label{expl:lts-comp}
  We can decompose the connector
  $L_R\colon\pow(\A\times(-))\to\pow(\B\times(-))$ from
  \autoref{expl:lts-connector} as follows. Define a further relational
  connector~$K_R\colon\pow(\A\times(-))\to\pow(\B\times(-))$ similarly
  as~$L_R$ but omit one of the directions, putting $S\mathrel{K_Rr}T$
  (for $S\in\pow(\A\times X)$, $T\in\pow(\B\times Y)$, and
  $r\colon\frel{X}{Y}$) iff for all $(l,m)\in R$ and $(l,x)\in S$,
  there is $(m,y)\in T$ such that $x\mathrel{r}y$. While $L_R$ has the
  feel of inducing a notion of heterogeneous bisimilarity (this will
  be made formal in \autoref{sec:het-sim}),~$K_R$ has a flavour of
  similarity, including as it does only a `forth'-type condition.
  Clearly, we have
  \begin{equation*}
    L_R=K_R\cap \rev{K_{\rev R}}.
  \end{equation*}
  Given a further set~$\C$ of labels and a relation
  $Q\subseteq\B\times\C$, we have
  \begin{equation*}
    K_Q\cdot K_R= K_{Q\cdot R}\quad\text{and}\quad L_Q\cdot L_R\le L_{Q\cdot R}.
  \end{equation*}
  It is a fairly typical phenomenon in describing composites of
  relational connectors that upper bounds such as the above are often
  straightforward, while the converse inequalities are more elusive or
  fail to hold. When showing
  $ K_{Q\cdot R}\,r\subseteq (K_Q\cdot K_R)\,r$ for
  $r\colon\frel{X}{Z}$, one gets away with using the trivial
  factorization $r=s\cdot t$ given by $s=r$, $t=\Delta_X$, while for a
  full description of $L_Q\cdot L_R$, we need to use
  \autoref{thm:comp-couniv}. Specifically, for $S\in\pow(\A\times X)$,
  $U\in\pow(\B\times X)$, we have $S\mathrel{(L_Q\cdot L_R)r} U$
  iff~$S$ and~$U$ satisfy conditions \emph{forth} and \emph{back},
  where \emph{forth} is given as follows and \emph{back} is given
  symmetrically: Whenever $(l,m)\in R$ and $(l,x)\in S$, then there
  are $A\in\pow(X)$, $B\in\pow(Z)$ such that $A\times B\subseteq r$
  and $x\in A$, and moreover (i) for all $(l',m)\in R$, there is
  $x'\in A$ such that $(l',x')\in S$, and (ii) for all $(m,p)\in Q$,
  there is $z\in B$ such that $(p,z)\in U$.
\end{example}

\noindent A further straightforward way to obtain relational
connectors is to pull them back along natural transformations:
\begin{lemdefn}\label{lem:transform}
  Let $L\colon F\to G$ be a relational connector, and let
  $\alpha\colon F'\Rightarrow F$, $\beta\colon G'\Rightarrow G$ be natural
  transformations. Then we have relational connectors
  $L\natcomp\alpha\colon F'\to G$,
  $\rev{\beta}\natcomp L\colon F\to G'$ defined on
  $r\colon\frel{X}{Y}$ by $(L\natcomp\alpha)r=Lr\cdot\alpha_X$ and
  $(\rev{\beta}\natcomp L)r= \rev{(\beta_Y)}\cdot Lr$, respectively.
\end{lemdefn}
\noindent In particular, from $\alpha\colon F\to G$, we always obtain
a relational connector $\alpha\natcomp\idcon{G}\colon F\to G$, which
plays a distinguished role:
\begin{defn}
  A relational connector $L\colon F\to G$ \emph{extends} a natural
  transformation $\alpha\colon F\to G$ if $\alpha_X\le L\Delta_X$ for
  all~$X$.
\end{defn}
\noindent (In particular, $L\colon F\to F$ extends~$F$ iff $L$
extends~$\id_F$.)
\begin{theorem}\label{thm:least-trans}
  Let $\alpha\colon F\to G$ be a natural transformation. The
  relational connector $\idcon{G}\natcomp\alpha$ is the least
  relational connector that extends~$\alpha$. In particular,
  $\idcon{G}$ is the least relational connector that extends~$G$.
\end{theorem}

\begin{example}\label{expl:trace-con}
  We have a variant~$\forthcon$ of the Barr extension of the functor
  $F=\pow(\A\times(-))$ modelling $\A$-LTS (\autoref{expl:LTS-lax})
  given by including only the \emph{forth} condition: For
  $r\colon\frel{X}{Y}$, $S\in FX$, $T\in FY$, we put
  $S\mathrel{\forthcon r} T$ iff for all $(l,x)\in S$, there is
  $(l,y)\in T$ such that $x\mathrel{r} y$. Now let
  $\iota\colon\A\times(-)\Rightarrow F$ be the inclusion natural
  transformation. Then we have a relational connector
  $\tracecon=\forthcon\natcomp\iota\colon\A\times(-)\to F$;
  explicitly, for $r\colon\frel{X}{Y}$, $(l,x)\in\A\times X$, and
  $T\in FY$, we have $(l,x)\mathrel{\tracecon r}T$ iff there exists
  $(l,y)\in T$ such that $x\mathrel{r}y$. By itself,~$\tracecon$ is
  not yet very interesting, but we can build further relational
  connectors using the constructions introduced above; for instance,
  we have a relational connector
  $\tracecon\cdot\rev{{\tracecon}}\colon F\to F$, described by
  $S\mathrel{(\tracecon\cdot\rev{{\tracecon}})r} T$ iff there exist
  $(a,x)\in S$, $(a,y)\in T$ such that $x\mathrel{r}y$; this connector
  is symmetric and extends~$F$ but fails to be transitive, hence is
  not a lax extension. We will later employ
  $\tracecon\cdot\rev{{\tracecon}}$ to relate LTS that share an infinite
  trace (\autoref{expl:trace-bisim}).
\end{example}

\begin{example}\label{expl:different-labels-transformation}
  Consider again the functors $F=\pow(\A\times(-))$ and
  $G=\pow(\B\times(-))$ together with a fixed relation on labels
  $R \subseteq \A \times \B$.  Note that, for every set $X$, the
  elements of $F X$ and $G X$ can be interpreted as relations
  $\frel{\A}{X}$ and $\frel{\B}{X}$, respectively.  Define the natural
  transformation $\alpha \colon F \Rightarrow G$ by
  $\alpha_X(S) = S \cdot \rev{R}$.
  Let $\forthcon^G \colon G \rightarrow G$ be the `forth' relational
  connector from \autoref{expl:trace-con} instantiated to~$G$, and
  consider the relational connector $\forthcon^G \natcomp \alpha$. For
  $S \in FX$, $T \in GY$ and $r \colon \frel{X}{Y}$, we have
  $S \mathrel{(\forthcon^G \natcomp \alpha) r} T$ iff
  $(S \cdot \rev{R}) \mathrel{\forthcon^G r} T$. Explicitly, the
  latter means that if $(l,m) \in R$ and $(l, x) \in S$, then there is
  $y \in Y$ such that $(m,y) \in T$ and $x\mathrel{r} y$. This
  coincides with the relational connector $K_R$ from
  \autoref{expl:lts-comp}, which is hence induced by a natural
  transformation and a lax extension. (It does not seem to be the case
  that $L_R$ as per
  \autoref{expl:lts-connector}/\autoref{expl:lts-comp} is induced in
  this way.)

We can instead compose with a natural transformation on the other
side. Let $\beta \colon G \Rightarrow F$ be given by
$\beta_X(T) = T \cdot R$, and let $\forthcon^F \colon G \rightarrow G$
be the connector $\forthcon^F$ from \autoref{expl:trace-con},
instantiated to~$F$.  The connector
$\rev{\beta} \natcomp \forthcon^F \colon F \rightarrow G$ is given,
for $S \in FX$, $T \in GY$ and $r \colon \frel{X}{Y}$, by
$S \mathrel{(\rev{\beta} \natcomp \forthcon^F) r} T$ iff
$S \mathrel{\forthcon^F r} (T \cdot R)$. Hence,
\begin{equation*}
  S \mathrel{(\rev{\beta} \natcomp \forthcon^F) r} T \iff \forall (l,x) \in S. \; \exists (l,m)\in R. \; (m,y) \in T\text{ and } x\mathrel{r}y,
\end{equation*}		
which differs from $K_R$ in that here, the quantification over $R$
is existential.
\end{example}

\begin{remark}\label{rem:lifting}
  Analogously to the fact that lax extensions of a functor~$F$ are
  equivalent to certain liftings of $F$ to the category of preordered
  sets~\cite{GoncharovEA23}, relational connectors $F\to G$ can be
  identified with certain liftings of
  $F \times G \colon \Set^2 \to \Set^2$ to the category of binary
  relations and relation-preserving pairs of functions. Indeed, this
  category is a fibration over $\Set^2$, and the relational connectors
  are precisely the liftings that preserve cartesian morphisms; a
  condition that has featured in situations where liftings of a
  functor~$F$ are used to derive notions of ``behavioural
  conformance'' for $F$-coalgebras (e.g. \cite{BaldanEA18,DBLP:journals/mscs/HasuoKC18,FGH+23,TBK+23}).
\end{remark}

\subsection{Lax Extensions as Relational Connectors}\label{sec:lax}

\noindent For context, we briefly discuss how the additional
properties of lax extensions are phrased in terms of the constructions
from \autoref{sec:rel-conn}, and in particular how lax extensions
relate to identity relational connectors.
\begin{defn}
  A relational connector $L\colon F\to F$ is \emph{transitive} if
  $L\cdot L\le L$, and \emph{symmetric} if $\rev{L}\le L$. Moreover,
  $L$ \emph{extends~$F$} if $\Delta_{FX}\subseteq L\Delta_X$ for all~$X$.
\end{defn}
\noindent The following observations are straightforward.
\begin{lemma}\label{lem:symm}
  Let $L\colon F\to F$ be a relational connector. Then~$L$ is
  symmetric iff $\rev{L}= L$ iff $L\le\rev{L}$.
\end{lemma}
\begin{lemma}\label{lem:lax-relconn}
  Let~$L\colon F\to F$ be a relational connector. Then the following
  hold.
  \begin{enumerate}
  \item\label{item:L2} $L$ satisfies condition~(L2) in the definition of lax
    extension iff~$L$ is transitive.
  \item\label{item:L3} $L$ satisfies condition~(L3) in the definition of lax
    extension iff~$L$ extends~$F$.
  \item \label{item:converse} $L$ preserves converse iff~$L$ is
    symmetric.
  \item\label{item:lax} $L$ is a lax extension of~$F$ iff~$L$ is
    transitive and extends~$F$.
  \item\label{item:L3-comp} If~$L$ extends~$F$, then
    $L\subseteq L\cdot L$.
  \item\label{item:lax-idempotent} If~$L$ is a lax extension, then~$L$
    is idempotent, i.e.\ $L\cdot L=L$.
  \end{enumerate}
\end{lemma}
\noindent Since lax extensions satisfy naturality, this implies
\begin{theorem}
  The lax extensions of a set functor~$F$ are precisely the transitive
  relational connectors that extend~$F$.
\end{theorem}
As indicated above, a special role is played by identity relational
connectors:
\begin{theorem}\label{thm:id-least}
  Let~$F$ be a set functor. Then,~$\idcon{F}$ is a symmetric lax
  extension of~$F$. Moreover,~$F$ has a
  diagonal-preserving lax extension iff~$\idcon{F}$ preserves
  diagonals.
\end{theorem}
\begin{proof}[sketch]
  Most subclaims are obvious by \autoref{lem:lax-relconn}
  and~\eqref{eq:idcon-diag}. To see that~$\idcon{F}$ is symmetric,
  show that~$\rev{(\idcon{F})}$ is a right identity: For
  $L\colon F\to G$, we have
  $L\cdot\rev{(\idcon{F})}=\rev{(\idcon{F}\cdot
    \rev{L})}=\rev{(\rev{L})}=L$ (using \autoref{lem:rev-comp}). \qed
\end{proof}
In connection with \autoref{thm:least-trans}, we obtain moreover:
\begin{corollary}
  The identity relational connector~$\idcon{F}$ is both the smallest
  lax extension and the smallest symmetric lax extension of a set
  functor~$F$.
\end{corollary}

\begin{example}\label{expl:Barr}
  If~$F$ preserves weak pullbacks, then~$\idcon{F}$ is the Barr
  extension of~$F$ (cf.\ \autoref{rem:barr}); this is immediate from
  \autoref{thm:id-least}, as one shows easily that the Barr extension
  is below every converse-preserving lax extension. For instance, the
  standard Egli-Milner lifting is an identity relational connector.
\end{example}

\section{Heterogeneous (Bi)simulations}\label{sec:het-sim}

\noindent We proceed to introduce a notion of heterogeneous
(bi)simulations relating systems of different type; we induce such
notions from relational connectors.

\begin{defn}
  Let $L\colon F\to G$ be a relational connector. A
  relation~$r\colon\frel{C}{D}$ is an \emph{$L$-simulation} between an
  $F$-coalgebra $(C,\gamma)$ and a $G$-coalgebra $(D,\delta)$ if
  
  \begin{equation*}
    \text{whenever $x\mathrel{r}y$, then }\gamma(x)\mathrel{Lr}\delta(y);
  \end{equation*}
  in pointfree notation, this means that
  $r\subseteq \rev{\delta}\cdot Lr\cdot \gamma$, equivalently
  $\delta\cdot r\subseteq Lr\cdot\gamma$.  States $x\in C$, $y\in D$
  are \emph{$L$-similar} if there exists an $L$-simulation~$r$ such
  that $x\mathrel{r}y$, in which case we write
  $x\simul_Ly$. Occasionally, we will designate the ambient coalgebras
  $C,D$ explicitly by writing $x\simul_L^{C,D}y$; thus,
  $\simul_L^{C,D}$ is a relation $\frel{C}{D}$.

  In case $F=G$,~$r$ is an \emph{$L$-bisimulation} if~$r$
  and~$\rev{r}$ are $L$-simulations. Correspondingly, states $x\in C$,
  $y\in D$ are \emph{$L$-bisimilar} if there exists an
  $L$-bisimulation~$r$ such that $x\mathrel{r}y$, in which case we
  write $x\bisim_Ly$ or, more explicitly, $x\bisim_L^{C,D}y$.
\end{defn}
We note that in case~$L$ is a lax extension, these definitions match
existing terminology (e.g.~\cite{MartiVenema15}). Monotonicity of
relational connectors ensures that by the Knaster-Tarski
theorem,~$\simul_L$ is the greatest fixpoint of the map taking~$r$ to
$\rev\delta\cdot Lr\cdot\gamma$, and in particular is itself an
$L$-simulation, correspondingly for~$\bisim_L$. We note that
$L$-similarity is invariant under coalgebra morphisms
(\autoref{sec:prelims}), a key fact that hinges on monotonicity and
naturality of relational connectors, lending further support to our
choice of axiomatics:
\begin{lemma}\label{lem:sim-mor}
  Let $L\colon F\to G$ be a  connector, let
  $r\colon\frel{C}{D}$ be an $L$-simulation between an $F$-coalgebra
  $(C,\gamma)$ and a $G$-coalgebra $(D,\delta)$, and let
  $f\colon (C',\gamma')\to (C,\gamma)$,
  $g\colon(C,\gamma)\to(C'',\gamma'')$ be $F$-coalgebra
  morphisms. Then $r\cdot f$ and $r\cdot\rev{g}$ are
  $L$-simulations. Symmetric properties hold for $G$-coalgebra
  morphisms. Thus, $L$-similarity is closed under behavioural
  equivalence (\autoref{sec:prelims}) on both sides.
\end{lemma}

\noindent Notions of (bi)simulation interact well with composition and
converse of relational connectors:

\begin{lemma}[Composites of simulations]
  \label{lem:functorial-simulation}
  Let $K\colon F\to G$ and $L\colon G\to H$ be relational connectors,
  and let $(C,\gamma)$ be an $F$-coalgebra, $(D,\delta)$ a
  $G$-coalgebra, and $(E,\varepsilon)$ an $H$-coalgebra. Then the
  composite $s\cdot r\colon\frel{C}{E}$ of a $K$-simulation
  $r\colon\frel{C}{D}$ and an~$L$-simulation $s\colon\frel{D}{E}$ is
  an $L\cdot K$-simulation. Thus,
  \begin{equation*}
    \simul^{D,E}_L\cdot\simul^{C,D}_K\;\,\subseteq\;\,\simul^{C,E}_{L\cdot K}
    \quad\text{and (if $F=G$)}\quad
    \bisim^{D,E}_L\cdot\bisim^{C,D}_K\;\,\subseteq\;\,\bisim^{C,E}_{L\cdot K}.
  \end{equation*}
\end{lemma}
\begin{lemma}[Converses of simulations]\label{lem:sim-conv}
  \label{item:sim-conv} Let~$L\colon F\to G$ be a relational
  connector, let~$(C,\gamma)$ be an $F$-coalgebra, and let
  $(D,\delta)$ be a $G$-coalgebra.  If~$r\colon \frel{C}{D}$ is an
  $L$-simulation, then~$\rev{r}\colon \frel{D}{C}$ is an
  $\rev{L}$-simulation. Thus,
\begin{equation*}
  \simul^{C,D}_{\rev{L}}\;\,=\;\,\rev{(\simul^{D,C}_L)}
  \quad\text{and (if $F=G$)}\quad
  \bisim^{C,D}_{\rev{L}}\;\,=\;\,\rev{(\bisim^{D,C}_L)}
\end{equation*}
\end{lemma}
\noindent It follows that notions of (bi)similarity inherit
properties expressed in terms of converse and composition from the
inducing lax extensions; for instance:
\begin{lemma}\label{lem:sim-props}
  Let~$L\colon F\to F$ be a relational connector. Then the following
  hold.
  \begin{enumerate}
  \item If~$L$ is transitive, then $\simul_L$ and $\bisim_L$ are
    transitive.
  \item If~$L$ is symmetric, then $\bisim_L$ is symmetric. Moreover,
    every $L$-simulation is an $L$-bisimulation, so
    $\simul_L\;\,=\;\,\bisim_L$.
  \item If~$L$ extends~$F$, then $\simul_L$ and $\bisim_L$ are
    reflexive. 
  \end{enumerate}
\end{lemma}

\noindent As a further immediate consequence of
\autoref{lem:functorial-simulation} and \autoref{lem:sim-conv}, we
have the following criterion for preservation of (bi)similarity under
relational connectors:
\begin{theorem}[Transfer of bisimilarity]\label{cor:bisim-pres}
  Let $K\colon F\to F$, $L\colon F\to G$, $H\colon G\to G$ be
  relational connectors such that $L\cdot K\cdot\rev{L}\le H$. Then
  $\simul_L\cdot \simul_K\cdot\rev{\simul_L}\;\,\subseteq\;\,
  \simul_H$ and
  $\bisim_L\cdot \bisim_K\cdot\rev{\bisim_L}\;\,\subseteq\;\,
  \bisim_H$.
\end{theorem}
\begin{example}[Transfer of bisimilarity between LTS of different
  type]\label{expl:bisim-transfer}
  Recall the relational connector
  $L_R\colon\pow(\A\times(-))\to\pow(\B\times(-))$ induced from a
  relation $R\colon\frel{\A}{\B}$ as per \autoref{expl:lts-connector}.
  We note that
  \begin{math}
    L_{\rev{R}}=\rev{(L_R)}.
  \end{math}
  This implies that for every $L$-simulation~$r$, $\rev{r}$ is an
  $L_{\rev{R}}$-simulation, so we suggestively write~$\bisim_R$ for
  $\simul_{L_R}$ and speak of \emph{$L_R$-bisimilarity}.

  Recall that the usual notion of bisimilarity on LTS is captured by
  the identity relational connectors on~$F$ and~$G$, respectively
  (\autoref{expl:LTS-lax}, \autoref{expl:Barr}). It is straightforward
  to check that if~$R$ is right total, then
  \begin{equation*}
    L_R\cdot\id_F\cdot \rev{L_R}= L_R\cdot \rev{L_R}\le\id_G,
  \end{equation*}
  so that by \autoref{cor:bisim-pres}, $\bisim_R$ transfers
  bisimilarity from $F$-coalgebras to~$G$-coalgebras. In elementwise
  notation, this is phrased as follows: Let $c,c'$ be states in an
  $F$-coalgebra~$C$, and let $d,d'$ be states in a $G$-coalgebra~$D$
  such that $c'\bisim_R d'$, $c\bisim_Rd$, and $c\bisim_F c'$.
  Then~$d\bisim_G d'$. Similarly, if~$R$ is left total, then
  $\bisim_R$ transfers bisimilarity from $G$-coalgebras to
  $F$-coalgebras, so of course if~$R$ is left and right total, then it
  transfers bisimilarity in both directions. A similar principle is
  under the hood of the proof of the operational equivalence of the
  standard $\lambda$-calculus and a variable-free variant called the
  algebraic $\lambda$-calculus in recent work on higher-order
  mathematical operational
  semantics~\cite{GoncharovEA24b}.
\end{example}

\begin{example}[Shared traces]\label{expl:trace-bisim}
  Recall the symmetric relational connector
  $\tracecon\cdot\rev{{\tracecon}}\colon F\to F$ from
  \autoref{expl:trace-con}, where $F=\pow(\A\times(-))$ is the functor
  modelling $\A$-LTS. States $x,y$ in $\A$-LTS are
  $\tracecon\cdot\rev{{\tracecon}}$-bisimilar iff~$x$ and~$y$ have a
  common infinite trace. We may view~$x$ as specifying a set of bad
  infinite traces; then~$x$ and~$y$ are \emph{not}
  $\tracecon\cdot\rev{{\tracecon}}$-bisimilar iff~$y$ does \emph{not} have
  a bad infinite trace.
\end{example}

\begin{example}[Weak simulation]
  Let $\A$ be a set of labels, with $\tau \in \A$ a distinguished
  label for ``internal'' steps.  Let $\A^*$ be the set of words over
  $\A$, with the empty word denoted by $\varepsilon$,
  $F = \pow(\A \times (-))$ and $G = \pow(\A^* \times (-))$.  We
  define a relational connector $L \colon F \rightarrow G$ by
  instantiating (the second half of)
  \autoref{expl:different-labels-transformation} to
  $R \subseteq \A \times \A^*$ given by
  $R = \{(l, \tau^i l \tau^j) \mid l \in \A, i,j \geq 0\} \cup
  \{(\tau, \varepsilon)\}$.  In the particular case where the
  transitions in the $G$-coalgebra $(D,\delta)$ at hand arise by
  composing transitions from an $F$-coalgebra $(D,\delta_0)$,
  $L$-simulations from an $F$-coalgebra $(C,\gamma)$ to $(D,\delta)$
  are precisely \emph{weak simulations} between the $\A$-LTS
  $(C,\gamma)$ and $(D,\delta_0)$.
\end{example}

\begin{example}[Conformance testing]\label{expl:ioco}
  In model-based testing, a \emph{specification} is compared to an
  \emph{implementation}.  Typically, both specifications and
  implementations are modelled as transition systems, and a given
  notion of \emph{conformance} stipulates when an implementation is
  correct w.r.t.\ a specification.  In the case of the \emph{ioco}
  (input/output conformance) relation~\cite{Tretmans08}, the
  specification is an LTS over a set of input and output labels. The
  implementation is an LTS as well, but is required to be
  \emph{input-enabled}, meaning that for every state and every input
  label there is an outgoing transition with that label. We focus on
  the deterministic case, which enables a coinductive formulation of
  ioco conformance~\cite{BosJM19}.  This example has been cast in a
  general coalgebraic framework~\cite{RotW23}, in which however the
  distinction between the type of specification and implementation
  cannot be made (and in fact, they are assumed to have the same state
  space).

  We write $X\to Y$ and $X \partialto Y$ for the sets of total and
  partial functions from~$X$ to~$Y$, respectively. We denote the
  domain of $f\colon X \partialto Y$ by $\dom(f) \subseteq X$, and put
  $X \partialtone Y=\{f\colon X \partialto Y\mid
  \dom(f)\neq\emptyset\}$. Now let $I,O$ be input and output alphabets,
  respectively.  Define the functor $F$ by
  $F(X) = (I \partialto X) \times (O \partialtone X)$, and the functor
  $G$ by $G(X) = (I \to X) \times (O \partialtone X)$. An
  $F$-coalgebra is a \emph{suspension automaton}, which is
  \emph{non-blocking} (there is always at least one output-labelled
  transition from every state). A $G$-coalgebra is an
  \emph{input-enabled} suspension automaton.

  Define $L \colon F \rightarrow G$ on  $r \colon \frel{X}{Y}$ by
  \begin{equation*}
    (\delta_I, \delta_O) \mathrel{L r} (\tau_I, \tau_O) \iff
    \begin{array}{l}
      \forall i \in \dom(\delta_I). \; \delta_I(i)\mathrel{r} \tau_I(i), \quad \text{and} \\
      \forall o \in \dom(\tau_O). \; o \in \dom(\delta_O) \text{ and } \delta_O(o)\mathrel{r} \tau_O(o). \\
    \end{array}
  \end{equation*}
  This is a relational connector, and $L$-simulations capture precisely the
  ioco-relation on suspension automata, in the coinductive formulation given
  in~\cite{BosJM19}. 
  
  The composite relational connector $\rev{L} \cdot L \colon F \to F$
  is described as follows:
  \begin{equation*}
  (\delta_I, \delta_O) \mathrel{(\rev{L} \cdot L) r} (\delta_I', \delta_O') \iff
  \begin{array}{l}
    \forall i \in \dom(\delta_I) \cap \dom(\delta_I'). \; \delta_I(i)\mathrel{r} \delta_I'(i), \quad \text{and} \\
    \exists o \in \dom(\delta_I) \cap \dom(\delta_I'). \; \delta_O(o)\mathrel{r} \delta_O'(o). \\
  \end{array}
  \end{equation*}
  The existential quantification on outputs arises in this factorization
  due to the non-emptyness of the domain of partial functions
  $O \partialtone X$. 
  Simulations for this composite relational connector are precisely
  the \emph{ioco compatibility} relations between specifications~\cite{BosJM19},
  generalized to a coalgebraic setting in~\cite{RotW23}.
\end{example}

\section{Kantorovich Relational Connectors}
\label{sec:kantorovich}
We next present a construction of relational connectors from relations
between modalities for the given functors; in honour of the formal
analogy with the classical Kantorovich metric and its coalgebraic
generalizations~\cite{BaldanEA18,WildSchroder22,DBLP:journals/logcom/SprungerKDH21}, 
we refer to the arising connectors as \emph{Kantorovich relational connectors}.

In this context, modalities are understood as induced by predicate
liftings in the style of coalgebraic
logic~\cite{Pattinson04,Schroder08}, and indeed we use the terms
\emph{modality} and \emph{predicate lifting} interchangeably. Recall
that an $n$-ary \emph{predicate lifting} for a functor~$F$ is a
natural transformation~$\lambda$ with components
\begin{equation*}
  \lambda_X\colon (2^X)^n\to 2^{FX}
\end{equation*}
(or just~$\lambda$) where $2^{(-)}$ denotes the \emph{contravariant
  powerset functor}; that is,~$2^X$ is the powerset of a set~$X$, and
$2^f\colon 2^Y\to 2^X$ takes preimages under a map $f\colon X\to
Y$. The naturality condition thus says explicitly that, for $a\in FX$
$f\colon X\to Y$, and $A_1,\dots,A_n\in 2^Y$, we have
$Ff(a)\in\lambda_Y(A_1,\dots,A_n)$ iff
$a\in\lambda_X(f^{-1}[A_1],\dots,f^{-1}[A_n])$. We say that~$\lambda$
is \emph{monotone} if
$\lambda(A_1,\dots,A_n)\subseteq \lambda(B_1,\dots,B_n)$ whenever
$A_i\subseteq B_i$ for $i=1,\dots,n$.  The \emph{dual} $\dual\lambda$
of $\lambda$ is the predicate lifting defined by
$\dual\lambda_X(A_1,\dots,A_n)=FX\setminus\lambda_X(X\setminus
A_1,\dots,X\setminus A_n)$.

In logical syntax, we abuse~$\lambda$ as an $n$-ary modality:
If~$\phi_1,\dots,\phi_n$ are formulae in some modal logic equipped
with a satisfaction relation~$\models$ between states in
$F$-coalgebras and formulae, with extensions
$\Sem{\phi_i}=\{x\in C\mid x\models\phi_i\}\in 2^C$ in a given
$F$-coalgebra $(C,\gamma)$, then the semantics of the modalized
formula $\lambda(\phi_1,\dots,\phi_n)$ is given by
$x\models\lambda(\phi_1,\dots,\phi_n)$ iff
$\gamma(x)\in\lambda_C(\Sem{\phi_1},\dots,\Sem{\phi_n})$. For
instance, the unary predicate lifting~$\Diamond$ for the powerset
functor~$\pow$ given by
$\Diamond_X(A)=\{S\in\pow(X)\mid S\cap A\neq\emptyset\}$ captures
precisely the usual diamond modality on Kripke frames (`there exists
some successor such that').

A set~$\Lambda$ of monotone predicate liftings for~$F$ induces a lax
extension $L_\Lambda$ of~$F$ defined for $r\colon\frel{X}{Y}$,
$a\in FX$, and $b\in FY$ by $a\mathrel{L_\Lambda r}b$ iff whenever
$a\in\lambda_X(A_1,\dots,A_n)$ for $n$-ary $\lambda\in\Lambda$ and
$A_1,\dots,A_n$, then $b\in\lambda_Y(r[A_1],\dots,r[A_n])$
(cf.~\cite{MartiVenema12,MartiVenema15,GorinSchroder13}). We show that
more generally, one can induce relational connectors from
\emph{relations} between predicate liftings:

\begin{defn}[Kantorovich connectors]\label{def:kantorovich}
  For a functor~$F$, we write $\predlift{F}$ for the set of monotone
  predicate liftings for~$F$.  Now let~$F$, $G$ be functors, and
  let~$\Lambda$ be a relation
  $\Lambda\colon\frel{\predlift{F}}{\predlift{G}}$ that
  \emph{preserves arity}; that is, if $(\lambda,\mu)\in\Lambda$,
  then~$\lambda$ and~$\mu$ have the same arity, which we then view as
  the \emph{arity} of $(\lambda,\mu)$. We define a relational
  connector $L_\Lambda\colon F\to G$ for $r\colon\frel{X}{Y}$,
  $a\in FX$, and $b\in GY$ by $a\mathrel{L_\Lambda r}b$ iff whenever
  $(\lambda,\mu)\in\Lambda$ is $n$-ary and $A_1,\dots,A_n\in 2^X$, then
  \begin{equation*}
    a\in\lambda_X(A_1,\dots,A_n)\quad\text{implies}
    \quad b\in\mu_Y(r[A_1],\dots,r[A_n]).
  \end{equation*}
  We briefly refer to $L_\Lambda$-similarity as
  \emph{$\Lambda$-similarity}, and write~$\simul_\Lambda$
  for~$\simul_{L_\Lambda}$. A relational connector~$L$ is
  \emph{Kantorovich} if it has the form $L=L_\Lambda$ for a
  suitable~$\Lambda$ as above. We write
  $\dual\Lambda=\{(\dual\lambda,\dual\mu)\mid(\lambda,\mu)\in\Lambda\}$.
\end{defn}

\begin{theorem}\label{thm:kantorovich}
  Under \autoref{def:kantorovich}, $L_\Lambda$ is indeed a relational
  connector.
\end{theorem}

\begin{example}\label{expl:kantorovich}%
  \begin{enumerate}[wide]
  \item\label{item:lts} For every~$l\in\A$, we have a predicate
    lifting $\Diamond_l$ for $\pow(\A\times(-))$ given by
    $\Diamond_l(A)=\{S\in\pow(\A\times X)\mid \exists x\in
    A.\,(l,x)\in S\}$. The arising modality is the usual diamond
    modality of Hennessy-Milner logic, and the dual of~$\Diamond_l$ is
    the usual box modality~$\Box_l$. The connectors
    $K_R,L_R\colon\pow(\A\times(-))\to\pow(\B\times(-))$ from
    \autoref{expl:lts-comp} are Kantorovich: We have $K_R=L_\Lambda$
    and $L_R=L_{\Lambda\cup\dual\Lambda}$ for
    $\Lambda=\{(\Diamond_l,\Diamond_m)\mid (l,m)\in R\}$.
  \item\label{item:trace-kant} We can restrict the predicate
    lifting~$\Diamond_l$ from the previous item to a predicate
    lifting~$\Diamond_l$ for $\A\times(-)$ (so
    $\Diamond_l(A)=\{(l,x)\mid x\in A\}$). The relational connector
    $\tracecon=\forthcon\natcomp\iota\colon\A\times(-)\to\pow(\A\times(-))$
    from \autoref{expl:trace-con} is Kantorovich for
    $\Lambda=\{(\Diamond_l,\Diamond_l)\mid l\in\A\}$. We will later
    give a Kantorovich description of the composite connector
    $\tracecon\cdot\rev{{\tracecon}}$ (\autoref{expl:kant-comp}).
  \item\label{item:ioco-kant} Given a label $l \in \A$,
    define the predicate lifting $\Diamond_l$ for $\A\partialto(-)$ by
    $\Diamond_l(A) = \{\delta\colon \A\partialto X \mid l \in
    \dom(\delta) \text{ and } \delta(l) \in A\}$ for $A\in 2^X$.  Its
    dual is given by
    $\Box_l(A) = \{\delta \mid l \in \dom(\delta) \text{ implies }
    \delta(l) \in A\}$.  Further, we define a $0$-ary modality
    ${\downarrow_l} = \{\delta \mid l \not \in \dom(\delta)\}$.
  These modalities allow us to capture the \emph{ioco} connector $L \colon F \rightarrow G$
  from \autoref{expl:ioco}. First, assuming that $I$ and $O$ are disjoint, the 
  modalities $\Diamond_l, \Box_l, \downarrow_l$ for $i \in I \cup O$ can be extended 
  to $F$ and $G$ in the obvious way by projection (and they are extended to total functions and partial functions 
  with a non-empty domain without change). Now, $L$ is Kantorovich for
  $\Lambda = \{(\Diamond_i, \Diamond_i) \mid i \in I\} \cup \{(\Box_o, \Box_o) \mid o \in O\} \cup \{({\downarrow_o}, {\downarrow_o}) \mid o \in O\}$.
    
\item \label{item:prob-kant} Given $\epsilon\in[0,1]$, we have
  predicate liftings~$L_{\epsilon,l}$ (for $l\in\A$) for the
  functor~$\Dist(\A\times(-))$ modelling probabilistic LTS, given by
  $L_{\epsilon,l}(A)=\{\alpha\in\Dist(\A\times X)\mid
  \alpha(\{l\}\times A)\ge\epsilon\}$ for $A\in 2^X$. Putting
  $\Lambda=\{(\Diamond_l,L_{\epsilon,l})\mid l\in\A\}$, we obtain
  relational connectors
  $L_\Lambda,L_{\dual\Lambda},L_{\Lambda\cup\dual\Lambda}\colon\pow(A\times(-))\to\Dist(\A\times(-))$.
  Explicitly, for $r\colon\frel{X}{Y}$, $S\in\pow(\A\times X)$, and
  $\alpha\in\Dist(\A\times Y)$, we have (i)
  $S\mathrel{L_\Lambda r}\alpha$ iff whenever $(l,x)\in S$, then
  $\alpha(\{l\}\times r[\{x\}])\ge\epsilon$ ; (ii)
  $S\mathrel{L_{\dual\Lambda}\,r}\alpha$ iff whenever
  $\alpha(\{l\}\times B)\ge\epsilon$, then there are $(l,x)\in S$ and
  $y\in B$ such that $x\mathrel{r}y$; and~(iii)
  $S\mathrel{L_{\Lambda\cup\dual\Lambda}\,r}\alpha$ iff both~(i)
  and~(ii) hold. Roughly speaking, similarity w.r.t.\ these connectors
  between an $\A$-LTS~$C$ and probabilistic $\A$-LTS~$D$ specifies
  what may happen in~$D$ with non-negligible probability,
  where~$\epsilon$ specifies the negligibility threshold. For
  instance, an $L_\Lambda$-simulation $r\colon C\to D$ witnesses that
  behaviour embodied in~$C$ is enabled with non-negligible probability
  in~$D$, while an $L_{\dual\Lambda}$-simulation $r\colon C\to D$
  witnesses that things that can happen with non-negligible
  probability in~$D$ are foreseen in~$C$.
  \end{enumerate}
\end{example}
\noindent We record basic facts on the interaction of the Kantorovich
construction with composition and converse of relational connectors:

\begin{theorem}\label{thm:kant-functorial}
  Let $\Lambda\colon\frel{\predlift{F}}{\predlift{G}}$ and
  $\Theta\colon\frel{\predlift{G}}{\predlift{H}}$. Then
  \begin{enumerate}
  \item
    \label{item:kant-comp} $L_\Theta\cdot L_\Lambda\le
    L_{\Theta\cdot\Lambda}$
  \item $\rev{(L_\Lambda)}=L_{\rev{(\dual\Lambda)}}$.
\end{enumerate}
\end{theorem}

\noindent (Recall that $\overline\Lambda$ dualizes all modalities.)
Specializing to relational connectors $F\to F$, we thus recover the
standard way of inducing lax extensions from predicate
liftings~\cite{MartiVenema15,GoncharovEA24} as described above:
\begin{corollary}\label{cor:lax-kant}
  Let~$F$ be a functor, and let
  $\Lambda\colon\frel{\predlift{F}}{\predlift{F}}$.
  \begin{enumerate}
  \item\label{item:kant-trans} If
    $\Lambda\cdot\Lambda\subseteq\Lambda$, then $L_\Lambda$ is
    transitive.
  \item\label{item:kant-symm} If~$\Lambda$ is closed under duals,
    i.e.\ $\dual\Lambda\subseteq\Lambda$ (equivalently
    $\dual\Lambda=\Lambda$), and symmetric, then~$L_\Lambda$ is
    symmetric.
  \item\label{item:kant-lax} If $\Lambda\subseteq\id$,
    then~$L_\Lambda$ is a lax extension of~$F$.
  \item\label{item:kant-lax-normal} If $\Lambda\subseteq\id$ and the
    set $\{\lambda\mid(\lambda,\lambda)\in\Lambda\}$ of predicate
    liftings is separating, then~$L_\Lambda$ is a normal lax extension
    of~$F$.
\end{enumerate}
\end{corollary}

\begin{remark}[Composing Kantororovich connectors]\label{rem:comp-kant}
  The upper bound $L_\Theta\cdot L_\Lambda\le L_{\Theta\cdot\Lambda}$
  on composites of Kantorovich connectors $L_\Theta, L_\Lambda$ given
  in \autoref{thm:kant-functorial} is not always tight. In the simple
  case of the connectors
  $K_R\colon\pow(\A\times(-))\to\pow(\B\times(-))$ induced by
  relations~$R\colon\frel{\A}{\B}$ (Examples~\ref{expl:lts-comp}
  and~\ref{expl:kantorovich}), we do indeed have exact equality
  (\autoref{expl:lts-comp}). For the general case, one can improve the
  upper bound (by including \emph{more} pairs of modalities) in at
  least two ways. First, in the composite $\Theta\cdot\Lambda$ of the
  relations $\Lambda\colon\frel{\predlift{F}}{\predlift{G}}$ and
  $\Theta\colon\frel{\predlift{G}}{\predlift{H}}$, one can include
  weakening in the middle step. Formally, we write~$\le$ for the
  pointwise inclusion order on predicate liftings, and put
  $\Theta\logcomp\Lambda=\{(\lambda,\pi)\mid\exists(\lambda,\mu)\in\Lambda,(\mu',\pi)\in\Theta\mid\mu\le\mu'\}$. Then
  $L_\Theta\cdot L_\Lambda\le L_{\Theta\logcomp\Lambda}$. Moreover,
  monotone predicate liftings are closed under taking positive Boolean
  combinations both above and below; e.g.\ if $\lambda$ and $\mu$ are
  unary monotone predicate liftings, then the transformation~$\pi$
  taking predicates $A,B$ to $\lambda(A\lor B)\land\mu(A\land B)$ is a
  binary monotone predicate lifting. We write $\Lambda^\posbool$ for
  the closure of~$\Lambda$ under componentwise positive Boolean
  combinations in this sense; e.g.\ if
  $(\lambda_1,\lambda_2),(\mu_1,\mu_2)\in\Lambda$, then
  $(\pi_1,\pi_2)\in\Lambda^\posbool$ where
  $\pi_i(A,B)=\lambda_i(A\lor B)\land\mu_i(A\land B)$. One checks
  easily that $L_\Lambda=L_{\Lambda^\posbool}$, so overall we have
  \begin{equation}\label{eq:comp-bound}
    L_\Theta\cdot L_\Lambda\le L_{\Theta^\posbool\,\logcomp\Lambda^\posbool}.
  \end{equation}
  We next give an example where one actually has equality; we leave it
  as an open problem whether equality holds in general.
\end{remark}
\begin{example}\label{expl:kant-comp}
  Recall from \autoref{expl:kantorovich}.\ref{item:trace-kant} that
  the connector $\tracecon\colon\A\times(-)\to\pow(\A\times(-))$
  equals $L_\Lambda$ where
  $\Lambda=\{(\Diamond_l,\Diamond_l)\mid l\in\A\}$; thus,
  $\rev{\tracecon}=L_{\rev{(\dual\Lambda)}}$ by
  \autoref{thm:kant-functorial}. Assume for simplicity that~$\A$ is
  finite. Note that we have
  \begin{equation*}\textstyle
    (\Land_{l\in\A}\Box_l(-)_l,\Lor_{l\in\A}\Diamond_l(-)_l)\in\Lambda^{\posbool}\logcomp\rev{(\dual\Lambda^{\,\posbool})},
  \end{equation*}
  where $\Land_{l\in\A}\Box_l(-)_l$ takes an $\A$-indexed family of
  predicates $A_l$ to $\bigcap_{l\in\A}\Box_lA_l$, correspondingly for
  $\Lor_{l\in\A}\Diamond_l(-)_l$, since this pair represents a valid
  implication over $\A\times(-)$. From this observation, one easily
  concludes that
  $\tracecon\cdot\rev{{\tracecon}}=L_\Lambda\cdot
  L_{\rev{(\dual\Lambda)}}=L_{\Lambda^{\posbool}\logcomp\rev{(\dual\Lambda^{\,\posbool})}}$,
  i.e.\ we have equality in the applicable instance
  of~\eqref{eq:comp-bound}. We will use this fact to obtain a logical
  characterization of $\tracecon\cdot\rev{{\tracecon}}$-bisimilarity
  (i.e.\ of sharing an infinite trace) in
  \autoref{expl:hm-expls}.
\end{example}

\begin{remark}
  Every lax extension of a finitary functor is induced by monotone
  predicate liftings as described
  above~\cite{KurzLeal12,MartiVenema12,MartiVenema15}. We leave it as an open
  problem whether every relational connector among finitary functors
  is Kantorovich.
\end{remark}
\section{Expressiveness}\label{sec:expressiveness}
\noindent We now go on to establish an expressiveness theorem in the
style of the classical Hennessy-Milner theorem, which states that two
states in finitely branching LTS are bisimilar iff they satisfy the
same formulae of Hennessy-Milner logic. Our present version subsumes
the classical theorem and coalgebraic generalizations, but also
variants for asymmetric comparisons such as simulation, and hence
instead works with forward preservation of formula satisfaction in a
logic with only positive Boolean combinations, introduced next:
\begin{defn}
  Let $\Lambda\colon\frel{\predlift{F}}{\predlift{G}}$. Then the
  set $\FLang(\Lambda)$ of \emph{$\Lambda$-formulae} $\phi,\psi$ is
  given by the grammar
  \begin{equation*}
    \FLang(\Lambda)\owns\phi,\psi::=\bot\mid\top\mid\phi\land\psi\mid\phi\lor\psi
    \mid\doublemod{\lambda,\mu}\phi\qquad((\lambda,\mu)\in\Lambda).
  \end{equation*}
  We interpret $\Lambda$-formulae over both $F$-coalgebras and
  $G$-coalgebras. For a state~$x$ in an $F$-coalgebra $(C,\gamma)$ and
  a $\Lambda$-formula~$\phi$, we write $x\models_F\phi$, or just
  $x\models\phi$, to indicate that $x$ \emph{satisfies}~$\phi$;
  similarly, we write $y\models_G\phi$ or just $y\models\phi$ to
  indicate that a state~$y$ in a $G$-coalgebra $(D,\delta)$
  \emph{satisfies}~$\phi$. We denote the \emph{extension} of~$\phi$
  in~$C$ by $\Sem{\phi}_C=\{x\in C\mid x\models C\}$, similarly
  for~$D$. The satisfaction relations are then defined by the usual
  clauses for the Boolean operators, and by
  \begin{align*}
    x\models_F\doublemod{\lambda,\mu}\phi \text{ iff } \gamma(x)\in\lambda(\Sem{\phi}_C),\qquad
    y\models_G\doublemod{\lambda,\mu}\phi  \text{ iff } \delta(y)\in\mu(\Sem{\phi}_D).
  \end{align*}
  We refer to the modal logic thus defined as $\Lang(\Lambda)$.
\end{defn}
One shows easily that the logic $\Lang(\Lambda)$ is preserved under
$L_\Lambda$-similarity:
\begin{proposition}\label{prop:adequacy}
  Let $\Lambda\colon\frel{\predlift{F}}{\predlift{G}}$, and
  let~$\phi$ be a $\Lambda$-formula. Whenever $x\simul_\Lambda y$ and
  $x\models_F\phi$, then $y\models_G\phi$.
\end{proposition}
\noindent The converse is less straightforward, and (like the
classical Hennessy-Milner theorem) depends on finite
branching. 
For brevity, we say that an $F$-coalgebra $(C,\gamma)$ is
\emph{finitely branching} if for every~$x\in C$, there exists a finite
subset $C'\subseteq C$ such that $\gamma(x)\in FC'\subseteq FC$ (cf.\
assumptions made in \autoref{sec:prelims}).
\begin{theorem}[Expressiveness]\label{thm:expr}
  Let $\Lambda\colon\frel{\predlift{F}}{\predlift{G}}$. Then
  $\Lambda$-similarity coincides with theory inclusion in
  $\Lang(\Lambda)$ on finitely branching coalgebras; that is, for
  states $x\in C$, $y\in D$ in finitely branching coalgebras
  $(C,\gamma\colon C\to FC)$ and $(D,\delta\colon D\to GD)$, we have
  $x\simul_\Lambda y$ iff for every $\Lambda$-formula~$\phi$, whenever
  $x\models_F\phi$, then $y\models_G\phi$.
\end{theorem}
\begin{proof}[sketch]
  Show that theory inclusion
  $r=\{(x,y)\in C\times D\mid\forall
  \phi\in\FLang(\Lambda).\,x\models_F\phi\implies y\models_G\phi\}$ is
  an $L_\Lambda$-simulation. \qed
\end{proof}

\begin{remark}
  From \autoref{thm:expr}, we recover in particular the coalgebraic
  generalization of the Hennessy-Milner theorem for behavioural
  equivalence~\cite{Pattinson04,Schroder08}, restricted to monotone
  modalities, by instantiating to $\Lambda\subseteq\id$ satisfying the
  usual separation condition (cf.~\autoref{cor:lax-kant}). This
  theorem applies to a logic with full Boolean propositional base;
  note here that when~$\Lambda$ is closed under duals, our logic
  admits an encoding of negation via negation normal forms. Also,
  \autoref{thm:expr} subsumes coalgebraic Hennessy-Milner theorems for
  behavioural preorders such as
  simulation~\cite{KapulkinEA12,WildSchroder21}. Our main interest is
  in heterogeneous examples, listed next.
\end{remark}

\begin{example}\label{expl:hm-expls}
  \begin{enumerate}[wide]
  \item From the Kantorovich description of the relational connectors
    $K_R,L_R\colon\pow(\A\times(-))\to\pow(\B\times(-))$ induced from
    $R\colon\frel{\A}{\B}$
    (\autoref{expl:kantorovich}.\ref{item:lts}), we obtain logical
    characterizations of $K_R$-similarity and $L_R$-(bi)similarity on
    finitely branching $\A$-LTS and $\B$-LTS. For instance, states
    $x\in C$, $y\in D$ in an $\A$-LTS~$C$ and a $\B$-LTS~$D$, both
    finitely branching, are $L_R$-bisimilar iff~$x$ and~$y$ satisfy
    the same formulae in a modal logic with modalities
    $\doublemod{\Diamond_l,\Diamond_m}$ and
    $\doublemod{\Box_l,\Box_m}$ for $(l,m)\in R$.
  \item In \autoref{expl:kantorovich}.\ref{item:ioco-kant}, a Kantorovich description 
  is given for \emph{ioco} simulation, yielding a logical characterization
  by \autoref{thm:expr}. The logic features the modalities $\Diamond_i$ for inputs $i \in I$,
  $\Box_o$ for outputs $o \in O$ and ``undefinedness'' modalities $\downarrow_o$,
  which hold at a state iff there is no outgoing $o$-transition from that state.
\item The Kantorovich definition of the relational connector
  $L_\Lambda\colon \pow(\A\times(-))\to\pow(\B\times(-))$, for
  $\Lambda=\{(\Diamond_l,L_{l,\epsilon})\mid l\in\A\}$ as per
  \autoref{expl:kantorovich}.\ref{item:prob-kant}, implies a logical
  characterization of $L_\Lambda$-simulation: Given states $x\in C$,
  $y\in D$ in an finitely branching $\A$-LTS~$C$ and a finitely
  branching probabilistic $\A$-LTS~$D$, we have $x\simul_\Lambda y$
  iff whenever~$x$ satisfies a formula~$\phi$ in the positive fragment
  of Hennessy-Milner logic with only diamond modalities~$\Diamond_l$,
  then~$y$ satisfies the probabilistic modal formula obtained
  from~$\phi$ by replacing $\Diamond_l$ with $L_{l,\epsilon}$
  throughout. Corresponding characterizations hold for
  $L_{\dual\Lambda}$-similarity and
  $L_{\Lambda\cup\dual\Lambda}$-similarity.
\item From the Kantorovich description of the connector
  $\tracecon\cdot\rev{{\tracecon}}\colon\pow(\A\times(-))\to\pow(\A\times(-))$,
  we obtain a logical characterization of
  $\tracecon\cdot\rev{{\tracecon}}$-bisimilarity, i.e.\ of sharing an
  infinite trace: States~$x,y$ in finitely branching $\A$-LTS are
  $\tracecon\cdot\rev{{\tracecon}}$-bisimilar iff whenever~$x$
  satisfies a formula~$\phi$ in a positive modal logic with $|\A|$-ary
  modalities $\Land_{l\in\A}\Box_l(-)_l$ as per
  \autoref{expl:kant-comp}, then~$y$ satisfies the formula obtained
  from~$\phi$ by replacing $\Land_{l\in\A}\Box_l(-)_l$ with
  $\Lor_{l\in\A}\Diamond_l(-)_l$ throughout. We note that in a
  scenario where we view~$x$ as specifying a set of bad traces, this
  means that the fact that~$y$ does \emph{not} have a bad trace can be
  witnessed by a single counterexample formula~$\phi$.
  \end{enumerate}
\end{example}

\section{Conclusions}

\noindent We have presented a systematic approach to comparing systems
of different transition types, abstracted as set functors in the
paradigm of universal coalgebra~\cite{Rutten00}: We induce notions of
\emph{heterogeneous (bi)simulation} from \emph{relational connectors}
among set functors. We have exhibited a number of key constructions of
relational connectors, including composition, converse, identity, and
a Kantorovich construction in which a connector is induced from a
relation between modalities. Building on the latter, we have proved a
Hennessy-Milner type theorem that characterizes heterogeneous
(bi)similarity in terms of theory inclusion in a flavour of positive
coalgebraic modal logic~\cite{KapulkinEA12} that is interpretable over
both of the involved system types. One instance of this result asserts
that absence of a shared trace between LTS can be witnessed by a pair
of modal formulae in Hennessy-Milner logic.

We leave quite a few problems open for further investigation, maybe
most notably including the question whether every relational connector
among finitary functors is Kantorovich (this holds for lax
extensions~\cite{KurzLeal12,MartiVenema12,MartiVenema15}, which form a
special case of relational connectors, and generalizes to arbitrary
functors when infinitary modalities are
allowed~\cite{Schroder08,GoncharovEA24}). More specifically, one would
be interested in a logical descriptions of composites of Kantorovich
connectors, working from \autoref{rem:comp-kant}. A further open
question is under what conditions similarity for a composite
relational connector $L \cdot K$ equals the composite of the
similarity relations for $L$ and $K$ respectively---currently, we only
have one inclusion (\autoref{lem:functorial-simulation}).  This is of
particular interest for the example on ioco conformance
(\autoref{expl:ioco}), where two specifications are known to be
compatible iff they have a common conforming
implementation~\cite{BosJM19}, a result that has been recovered in a
coalgebraic setting~\cite{RotW23}. A further issue for future research
is to develop the coinductive up-to
techniques~\cite{DBLP:books/cu/12/PousS12,DBLP:journals/acta/BonchiPPR17}
for relational connectors.  \bibliographystyle{splncs04}
\bibliography{coalgml}
\clearpage
\appendix

\section{Omitted Details and Proofs}

\subsection{Details for \autoref{sec:rel-conn}}
\subsection*{Details for \autoref{expl:lts-connector}}

We have to show that $L_R$ is actually a relational connector. We
could prove this using \autoref{thm:kantorovich} but we give a direct
proof to avoid the forward reference. Monotonicity is clear; we prove
naturality. So let $f\colon X'\to X$, $g\colon Y'\to Y$, and
$r\colon\frel{X}{Y}$; we have to show that
\begin{equation*}
  L_R(\rev{g}\cdot r\cdot f)=\rev{(Gg)}\cdot L_Rr\cdot Ff.
\end{equation*}
We prove the two inclusions separately;

$\subseteq$: Let $S\in\pow(\A\times X')$, $T\in\pow(B\times Y')$ such
that $S\mathrel{L_R(\rev{g}\cdot r\cdot f)}T$. We have to show that
$\pow(\A\times f)(S)\mathrel{L_Rr}\pow(\B\times g)(T)$. So let
$(a,b)\in R$ and $(a,x)\in \pow(\A\times f)(S)$. Then there is
$(a,x')\in S$ such that $f(x')=x$. Thus, there is $(b,y')\in T$ such
that $x'\mathrel{(\rev{g}\cdot r\cdot f)}y'$, i.e.\
$x=f(x')\mathrel{r}g(y')$, and $(b,g(y'))\in\pow(\B\times g)(T)$ as
required. The remaining condition is shown symmetrically.

$\supseteq$: Let $S\in\pow(\A\times X')$, $T\in\pow(B\times Y')$ such
that $\pow(\A\times f)(S)\mathrel{L_Rr}\pow(\B\times g)(T)$. We have
to show that $S\mathrel{L_R(\rev{g}\cdot r\cdot f)}T$. So let
$(a,b)\in R$ and $(a,x')\in S$. Then
$(a,f(x'))\in\pow(\A\times f)(S)$, so there is
$(b,y)\in\pow(\B\times g)(T)$ such that $f(x')\mathrel{r} y$. Now~$y$
has the form $y=g(y')$ for some $y'$ such that $(b,y')\in T$, and then
$x'\mathrel{(\rev{g}\cdot r\cdot f)}y'$ as required. Again, the remaining
condition is symmetric. \qed

\subsection*{Full proof of \autoref{lem:comp}}
\emph{Monotonicity:} Let $r\subseteq r'\colon \frel{X}{Z}$, and let
$a\;(L\cdot K)r\;c$; we have to show that $a\;(L\cdot K)r'\;c$.  By
definition, we have $r=s\cdot t$ for some $t\colon\frel{X}{Y}$,
$s\colon\frel{Y}{Z}$, and~$b\in GY$ such that $a\;Kt\;b\;Ls\;c$. Let
$Y'=Y+(r'\setminus r)$, w.l.o.g.\ with the coproduct just being a
disjoint union, let $i\colon Y\to Y'$ be the left coproduct injection,
and define $t'\colon \frel{X}{Y'}$, $s'\colon \frel{Y'}{Z}$ by
\begin{align*}
  t'&=t\cup\{(x,(x,z))\mid (x,z)\in r'\setminus r\}\\
  s'&=s\cup\{((x,z),z)\mid (x,z)\in r'\setminus r\}.
\end{align*}
Then $r'=s'\cdot t'$, $t=\rev{i}\cdot t'$, and $s=s'\cdot i$. By
naturality, $Kt=\rev{(Gi)}\cdot Lt'$ and $Ls=Ls'\cdot Gi$. Thus,
$a\mathrel{Kt'} Gi(b)\mathrel{Ls'}c$, so $a\mathrel{(L\cdot K)r'}c$ as
required.

\emph{Naturality:} Let $r\colon\frel{X}{Z}$, $f\colon X'\to X$,
$g\colon Z'\to Z$; we have to show that
\begin{equation*}
  (L\cdot K)(\rev{g}\cdot r\cdot f)=\rev{(Hg)}\cdot(L\cdot K)r\cdot Ff.
\end{equation*}
We split this equality into two inclusions:

\emph{`$\supseteq$':} Let $a\in FX'$, $b\in HZ'$ such that
$a\mathrel{(\rev{(Hg)}\cdot (L\cdot K)r \cdot Ff)}c$, i.e.\ we have
$r=s\cdot t$ and~$b$ such that
$Ff(a)\mathrel{Kt}b\mathrel{Ls}Hg(c)$. We have to show that
$a\mathrel{(L\cdot K)(\rev{g}\cdot r\cdot f)}c$. Indeed, by naturality
of~$K$ and~$L$, we have
$a\mathrel{K(t\cdot f)}b\mathrel{L(\rev{g}\cdot s)} c$, and hence
$a\mathrel{(L\cdot K)(\rev{g}\cdot s\cdot t\cdot f)} c$. Since
$s\cdot t=r$, this implies the claim.

\emph{`$\subseteq$':} Let $a\in FX'$, $c\in HZ'$ such that
$a\mathrel{(L\cdot K)(\rev{g}\cdot r\cdot f)} c$, so we have
$\rev{g}\cdot r\cdot f=s\cdot t$ and~$b$ such that
$a\mathrel{Kt} b\mathrel{Ls} c$. We have to show that
$a\mathrel{(L\cdot K)(\rev{g}\cdot r\cdot f)} c$. Since
$\rev{f}\cdot f\supseteq\Delta_{X'}$ and
$\rev{g}\cdot g\supseteq\Delta_{Y'}$ (totality), we have
\begin{equation*}
  a\mathrel{K(t\cdot\rev{f}\cdot f)} b
  \mathrel{L(\rev{g}\cdot g\cdot s)} c
\end{equation*}
by monotonicity of~$K$ and~$L$. By naturality of~$K$ and~$L$, it
follows that
\begin{equation*}
  Ff(a)\mathrel{K(t\cdot\rev{f})} b
  \mathrel{L(g\cdot s)} Hg(c),
\end{equation*}
so $Ff(a)\mathrel{(L\cdot K)(g\cdot s\cdot t\cdot \rev{f})}
Hg(c)$. But we have
\begin{equation*}
  g\cdot s\cdot t\cdot \rev{f}=g\cdot\rev{g}\cdot r\cdot f\cdot\rev{f}\subseteq r
\end{equation*}
by univalence, so $Ff(a)\mathrel{(L\cdot K)r} Hg(c)$ follows
by monotonicity of $L\cdot K$, which we have already shown above. \qed

\subsection*{Proof of \autoref{lem:comp-subset}}

For $r\supseteq s\cdot t$, we have
$Ls\cdot Kt\subseteq (L\cdot K)(s\cdot t)\subseteq (L\cdot K)r$ by
monotonicity of $L\cdot K$. \qed

\subsection*{Proof of \autoref{lem:couniv-decomp}}
We define $f$ by $f(y)=(\rev{(t')}[\{y\}],s'[\{y\}])$; we have
$f(y)\in Y$ because $s'\cdot t'=r$. 
Let $x \in X$, $y \in Y'$ and $z \in Z$.
We check the requisite properties:

$t' = \rev{f} \cdot t$: 
By definition of $t$ and $f$, $x\mathrel{t'}y \iff x\in\rev{(t')}[\{y\}] \iff x\mathrel{t} f(y)$.

$s' = s \cdot f$:
By definition of $s$ and $f$, $y\mathrel{s'}z \iff z \in s'[\{y\}] \iff f(y)\mathrel{s}z$.
\qed

\subsection*{Proof of \autoref{thm:comp-couniv}}

Let~$Y$ be the intermediate set in the couniversal factorization, so
$s\colon\frel{Y}{Z}$ and $t\colon\frel{X}{Y}$.

The right-to-left inclusion is immediate from the definition of
$(L\cdot K)r$ as per~\eqref{eq:comp}. 
For the left-to-right inclusion, note that for every factorization $r=s'\cdot t'$ of~$r$ into $s'\colon\frel{Y'}{Z}$ and $t'\colon\frel{X}{Y'}$, by \autoref{lem:couniv-decomp}, there is $f\colon Y'\to Y$ such that $t' = \rev{f} \cdot t$ and $s' = s \cdot f$.
Therefore, by naturality and univalence, we have $Ls' \cdot K t' = L(s \cdot f) \cdot K(\rev{f} \cdot t) = Ls \cdot Gf \cdot \rev{(Gf)} \cdot Kt \leq Ls \cdot Kt$.
\qed

\subsection*{Proof of \autoref{lem:assoc}}
We have
\begin{align*}
  & ((M\cdot H)\cdot L)r\\
  & = \bigvee_{r=s\cdot t}(M\cdot H)s\cdot Lt\\
  & = \bigvee_{r=s\cdot t}\big(\bigvee_{s=u\cdot v}(M u\cdot Hv)\big)\cdot Lt\\
  & = \bigvee_{r=u\cdot v\cdot t}M u\cdot Hv\cdot Lt,
\end{align*}
which by an analogous computation equals $M\cdot (H\cdot L)r$. In
the last step, we have used join continuity of relational
composition. \qed

\subsection*{Proof of \autoref{lem:id}}

\emph{Monotonicity} is immediate; we show \emph{naturality}. So let
$r\colon\frel{X}{Y}$, $f\colon X'\to X$, $g\colon Y'\to Y$; we have
to show that
\begin{equation*}
  \idcon{F}(\rev{g}\cdot r\cdot f)=\rev{(Fg)}\cdot\idcon{F}r\cdot Ff.
\end{equation*}
We prove the two inclusions separately:

\emph{$\subseteq$:} Let
$b\mathrel{\idcon{F}(\rev{g}\cdot r\cdot f)}c$; we have to show that
$b\mathrel{(\rev{(Fg)}\cdot\idcon{F}r\cdot Ff)}c$, i.e.\
$Ff(b)\mathrel{\idcon{F}r}Fg(c)$. So let~$L\colon G\to F$ be a
relational connector, let $s\colon \frel{Z}{X}$, and let $a\in GZ$
such that $a\mathrel{Ls} Ff(b)$; we have to show
$a\mathrel{L(r\cdot s)} Fg(c)$. From $a\mathrel{Ls} Ff(b)$, we
obtain $a\mathrel{L(\rev{f}\cdot s)} b$ by naturality of~$L$, and
hence
\begin{equation*}
  a\mathrel{L(\rev{g}\cdot r\cdot f\cdot\rev{f}\cdot s)} c
\end{equation*}
because $b\mathrel{\idcon{F}(\rev{g}\cdot r\cdot f)}c$. Since
$f\cdot\rev{f}\subseteq\Delta_X$, we obtain
$a\mathrel{L(\rev{g}\cdot r\cdot s)} c$ by monotonicity, and hence
by naturality $a\mathrel{L(r\cdot s)} Fg(c)$, as required.

\emph{$\supseteq$:} Let $Ff(b)\mathrel{\idcon{F}r}Fg(c)$; we have to
show $b\mathrel{\idcon{F}(\rev{g}\cdot r\cdot f)}c$. So let
$L\colon G\to F$ be a relational connector, let
$s\colon\frel{Z}{X'}$, and let $a\in GZ$ such that $a\mathrel{Ls}b$;
we have to show that $a\mathrel{L(\rev{g}\cdot r\cdot f\cdot
  s)}c$. Since $\Delta_{X'}\subseteq\rev{f}\cdot f$ (totality),
$a\mathrel{Ls}b$ implies $a\mathrel{L(\rev{f}\cdot f\cdot s)}b$ by
monotonicity, and hence $a\mathrel{L(f\cdot s)} Ff(b)$ by
naturality. Since $Ff(b)\mathrel{\idcon{F}r}Fg(c)$, we thus obtain
$a\mathrel{L(r\cdot f\cdot s)} Fg(c)$, which by naturality implies
that $a\mathrel{L(\rev{g}\cdot r\cdot f\cdot s)} c$, as required. \qed

\subsection*{Proof of \autoref{lem:id-comp}}

We show that~$\idcon{F}$ is a left identity ($L=\idcon{F}\cdot L$). By
a symmetric argument, one obtains that composition of relational connectors also has
right identities~${\idcon{F}}'$, and then the left and right
identities are equal
($\idcon{F}=\idcon{F}\cdot {\idcon{F}}'={\idcon{F}}'$), so~$\idcon{F}$
is also a right identity.

We thus have to show that for $L\colon G\to F$, we have
\begin{equation*}
  \idcon{F}\cdot L=L.
\end{equation*}
We split this equality into two inclusions:

\emph{$\subseteq$:} Let $r\colon\frel{X}{Y}$, and let
$a\mathrel{(\idcon{F}\cdot L)r}c$; we have to show $a\mathrel{Lr}c$ By
definition of the composition of relational connectors, we have $r=s\cdot t$ and~$b$
such that
\begin{equation*}
  a\mathrel{Lt}b\mathrel{\idcon{F}s}c.
\end{equation*}
By definition of $\idcon{F}$, it follows that
$a\mathrel{L(s\cdot t)}c$, i.e.\ $a\mathrel{Lr}c$, as required.

\emph{$\supseteq$:} Let $r\colon\frel{X}{Y}$, and let
$a\mathrel{Lr} c$; we have to show that
$a\mathrel{(\idcon{F}\cdot L)r} c$. We have $r=\Delta_Y\cdot r$, and
by~\eqref{eq:idcon-diag},
\begin{equation*}
  a\mathrel{L} c \mathrel{\idcon{F}\Delta_Y} c,
\end{equation*}
so $a\mathrel{(\idcon{F}\cdot L)r} c$, as required. \qed

\subsection*{Proof of \autoref{lem:converse-meet-product}}

For naturality of the converse, we have
\begin{align*}
	\rev{L}(\rev{g} \cdot r \cdot f)
	&= \rev{(L(\rev{(\rev{g} \cdot r \cdot f)}))}\\
	&= \rev{(L(\rev{f} \cdot \rev{r} \cdot g))}\\
	&= \rev{(\rev{(Gf)} \cdot L(\rev{r}) \cdot Fg)} \\
	&= \rev{(Fg)} \cdot \rev{(L(\rev{r}))} \cdot Gf \\
	&= \rev{(Fg)} \cdot \rev{L}(r) \cdot Gf .
\end{align*}

Naturality of the meet is (also) straightforward, using that
$$
	\rev{(Gg)} \cdot (L \cap K)(r) \cdot Ff = (\rev{(Gg)} \cdot Lr \cdot Ff) \cap (\rev{(Gg)} \cdot Kr \cdot Ff)
$$
for \emph{functions} $f,g$.

For the product, given relational connectors $L_1 \colon F_1 \to G_1$
and $L_2 \colon F_2 \to G_2$,
note that the construction can be reformulated as 
$$
(L_1 \times L_2) r = (\rev{(\pi_{1,G})} \cdot (L_1r) \cdot \pi_{1,F}) \cap (\rev{(\pi_{2,G})} \cdot (L_2r) \cdot \pi_{2,F})
$$
where $\pi_{i,F} \colon F_1 \times F_2 \rightarrow F_i$ and $\pi_{i,G} \colon G_1 \times G_2 \rightarrow G_i$
are the projections. Since the meet of relational connectors is again a relational connector (shown above), the
result follows once we show that both
$(\rev{(\pi_{i,G})} \cdot (L_i r) \cdot \pi_{i,F})$ for $i \in \{1,2\}$ are relational connectors.
This, in turn, follows by \autoref{lem:transform} since $\pi_{i,F}$ and $\pi_{i,G}$ are natural transformations.
(To avoid the forward reference, it is straightforward to prove explicitly that 
each $(\rev{(\pi_{i,G})} \cdot (L_i r) \cdot \pi_{i,F})$ is a relational connector.)

\subsection*{Proof of \autoref{lem:rev-comp}}
All properties are straightforward; we prove only the last one: Let
$r\colon\frel{X}{Y}$. We have
\begin{align*}
  & \rev{(L\cdot K)}r \\
  & = \rev{\big(\textstyle\bigvee_{\rev{r}=s\cdot t}Ls\cdot Kt\big)}\\
  & = \rev{\big(\textstyle\bigvee_{r=\rev{t}\cdot \rev{s}}Ls\cdot Kt\big)}\\
  & = \rev{\big(\textstyle\bigvee_{r=t\cdot s}L(\rev s)\cdot K(\rev t)\big)}\\
  & = \big(\textstyle\bigvee_{r=t\cdot s}\rev{K(\rev{t})}\cdot \rev{L(\rev{s})}\big)\\
  & = (\rev K\cdot\rev L)r.
\end{align*}
\qed

\subsection*{Details for \autoref{expl:lts-comp}}

$K_Q\cdot K_R\le K_{Q\cdot R}$: Let $s\colon\frel{Y}{Z}$,
$t\colon\frel{X}{Y}$, $S\in\pow(\A\times X)$,
$T\in\pow(\B\times Y)$, $U\in\pow(\C\times Z)$ such that
$r=s\cdot t$, $S\mathrel{K_Rt}T$, $T\mathrel{K_Qs}U$. We have to
show that $S\mathrel{K_{Q\cdot R}(s\cdot t)} U$. So let
$(a,b)\in R$, $(b,c)\in Q$, and $(a,x)\in S$. We have to find
$(c,z)\in U$ such that $x\mathrel{(s\cdot t)}z$. Since
$S\mathrel{K_Rt}T$, there is $(b,y)\in T$ such that $x\mathrel{t}y$,
and since $T\mathrel{K_Qs}U$, there is $(c,z)\in U$ such that
$y\mathrel{s}z$. Then $x\mathrel{(s\cdot t)}z$, as required.

$L_Q\cdot L_R\le L_{Q\cdot R}$: By the description of $L_R$ in terms
of $K_R$, we have
\begin{align*}
  & L_Q\cdot L_R \\
  & = (K_Q\cap\rev{(K_{\rev Q})})\cdot (K_R\cap \rev{(K_{\rev R})})\\
  & \le (K_Q\cdot K_R)\cap(K_{\rev Q}\cdot K_{\rev R})\\
  & \le K_{Q\cdot R}\cap \rev{(K_{\rev R\cdot\rev Q})}\\
  & =K_{Q\cdot R}\cap \rev{(K_{\rev{(Q\cdot R)}})}\\
  & = L_{Q\cdot R}
\end{align*}
using monotonicity of composition in the second step and the previous
inequality in the third step.

$ K_{Q\cdot R}\le K_Q\cdot K_R$ when~$Q$ is left total and~$R$ is
right total: Let $S\in\pow(\A\times X)$, $U\in\pow(\C\times Z)$, and
$r\colon\frel{X}{Z}$ such that $S\mathrel{K_{Q\cdot R}r} U$. We have
to construct a set~$Y$, an element $T\in\pow(\B\times Y)$, and
relations $s\colon\frel{Y}{Z}$, $t\colon\frel{X}{Y}$ such that
$s\cdot t\subseteq r$, $S\mathrel{K_Rt}T$, and
$T\mathrel{K_Qs}U$. Indeed we can put $Y=X$, $t=\Delta_X$, $s=r$
(trivially ensuring $r=s\cdot t$), and
\begin{equation*}
  T=\{(b,x)\in\B\times X\mid \forall (b,c)\in Q.\,\exists (c,z)\in U.\,x\mathrel{r}z\}.
\end{equation*}
Then $T\mathrel{K_Qs}U$ by construction. To show that also
$S\mathrel{K_Rt}T$, let $(a,b)\in R$ and $(a,x)\in S$; we have to show
$(b,x)\in T$. So let $(b,c)\in Q$. Then $(a,c)\in Q\cdot R$, so since
$S\mathrel{K_{Q\cdot R}r} U$, there exists $(c,z)\in U$ such that
$x\mathrel{r} z$, implying $(b,x)\in T$.

\emph{Full description of $L_Q\cdot L_R$:} Again, we use
$L_R=K_R\cap\rev{K_{\rev R}}$. Let $r=s\cdot t$, with
$s\colon\frel{Y}{Z}$, $t\colon\frel{X}{Y}$, be the couniversal
factorization of~$r$ (\autoref{def:couniv-decomp}). We define
$T\in\pow(\B\times Y)$ by
\begin{align*}
  T=\{(m,(A,B))\mid (&\forall (l,m)\in R.\,\exists x\in A.\,(l,x)\in S)\;\land\\&\forall (m,p)\in Q.\,\exists z\in B.\,(p,z)\in U\}.
\end{align*}
Then $S\mathrel{\rev{K_{\rev R}}t}T$ and $T\mathrel{K_Q s} U$ by
construction, and~$T$ is the largest subset of $\B\times Y$ with this
property, so by \autoref{thm:comp-couniv}, we have
$S\mathrel{(L_Q\cdot L_R)r}U$ iff $S\mathrel{K_R t}T$ and
$T\mathrel{\rev{K_{\rev Q}}s} U$ (noting that $K_R t$ is upwards
closed in the right argument, and $\rev{K_{\rev Q}}s$ in the left
argument). We show that the former condition is equivalent to
\emph{forth}; ones shows symmetrically that the second condition is
equivalent to \emph{back}.

To this end, we just unfold the definition of $S\mathrel{K_R
  t}T$. This definition requires that for $(l,x)\in S$ and
$(l,m)\in R$, we have $(m,(A,B))\in T$ such that $x\mathrel{t}(A,B)$,
i.e.~$x\in A$.  Unfolding the definitions of $(m,(A,B))\in T$ and
$(A,B)\in Y$ gives exactly \emph{forth}.

\subsection*{Proof of \autoref{lem:transform}}

\noindent Put $L_{\alpha,\beta} = \rev{\beta}\natcomp L\natcomp\alpha$.
We show that $L_{\alpha,\beta}$ is a relational connector.
For $L\natcomp\alpha$ and $\rev{\beta}\natcomp L$ the claim then follows by replacing $\beta$ or $\alpha$ with the identity natural transformation, respectively.

Monotonicity is immediate; we show naturality: For
$r\colon X\to Y$, $f\colon X'\to X$, $g\colon Y'\to Y$, we have
\begin{align*}
  & L_{\alpha,\beta}(\rev g\cdot r\cdot f)\\
  & = \rev{\beta_{Y'}}\cdot (\rev g\cdot r\cdot f)\cdot \alpha_{X'} &&\by{definition}\\ 
  & = \rev{\beta_{Y'}}\cdot \rev{(Gg)}\cdot Lr\cdot Ff\cdot \alpha_{X'}&&\by{naturality of~$L$}\\
  & = \rev{(G'g)}\cdot\rev{\beta_Y}\cdot Lr\cdot \alpha_X\cdot F'f
  && \by{naturality of $\alpha,\beta$}\\
  & = \rev{(G'g)}\cdot L_{\alpha,\beta}r \cdot F'f&&\by{definition}.
\end{align*}

\subsection*{Proof of \autoref{thm:least-trans}}

\noindent We first check that $\idcon{G}\natcomp\alpha$
extends~$\alpha$: By~\eqref{eq:idcon-diag}, we have
$\alpha_X=\Delta_{GX}\cdot\alpha_x\le(\idcon{G}\Delta_X)\cdot\alpha_X=(\idcon{G}\natcomp\alpha)\Delta_X$.

Now let $L\colon F\to G$ be a relational connector that
extends~$\alpha$, and let $r\colon\frel{X}{Y}$. Then
\begin{equation*}
  Lr=(\idcon{G}\cdot L)r\ge\idcon{G}r\cdot L\Delta_X\ge
  \idcon{G}r\cdot\alpha_X=(\idcon{G}\natcomp\alpha)r. 
\end{equation*}
\qed

\subsection*{Details for \autoref{rem:lifting}}

Analogously to the fact that lax extensions of a functor
$F \colon \Set \to \Set$ can be thought of as certain liftings of $F$
along the forgetful functor $\ORD \to \Set$ \cite{GoncharovEA23} from
the category of preordered sets and monotone maps, relational
connectors from a functor $F \colon \Set \to \Set$ to a functor
$G \colon \Set \to \Set$ can be thought of as certain liftings of
$F \times G \colon \Set^2 \to \Set^2$ along the canonical forgetful
functor $U \colon \BREL \to \Set^2$ from the category of binary
relations and relation-preserving pairs of functions; i.e, a morphism
in $\BREL$ from a relation $r \colon\frel{X}{Y}$ to a relation
$s \colon \frel{X'}{Y'}$ is pair of functions
$(f \colon X \to X',g \colon Y \to Y')$ such that
$r \leq \rev{g} \cdot s \cdot f$.
Indeed, a lifting $L \colon \BREL \to \BREL$ of
$F \times G \colon \Set^2 \to \Set^2$ to $\BREL$ along $U$, in the
sense that $U \cdot L = (F \times G) \cdot U$, assigns to each
relation $r \colon \frel{X}{Y}$ a relation $L r \colon \frel{FX}{FY}$
such that for all relations $r \colon \frel{X}{Y}$ and
$s \colon \frel{X'}{Y'}$ and all functions $f \colon X \to X'$ and
$g \colon Y \to Y'$, whenever $r \leq \rev{g} \cdot s \cdot f$, then
$L r \leq \rev{(Gg)} \cdot Lr \cdot Ff$.  This condition is equivalent
to the following:
\begin{enumerate}
\item if $r \leq s$, then $Lr \leq Ls$, for all $r,s \colon \frel{X}{Y}$;
\item $L(\rev{g} \cdot s \cdot f) \leq \rev{G g} \cdot L s \cdot F f$, for every $s \colon \frel{X'}{Y'}$ and all functions $f \colon X \to X'$ and $g \colon Y \to Y'$.
\end{enumerate}
Therefore, the relational connectors from $F$ to $G$ define liftings
of $F \times G$ to $\BREL$ along $U$.  In fact, it is easy to see that
they correspond precisely to the liftings that additionally satisfy
the following condition: For all functions $f \colon X \to X'$ and
$g \colon Y \to Y'$, whenever $r = \rev{g} \cdot s \cdot f$, then
$L r = \rev{(Gg)} \cdot Lr \cdot Ff$.
In other words, relational connectors from $F$ to $G$ correspond
precisely to the liftings of $F \times G$ along $U$ that preserve
$U$-initial morphisms (also called cartesian or fibered liftings).
This is a very natural condition that is imposed frequently in
situations where liftings of a functor $F \colon \catA \to \catA$
along a functor $\catB \to \catA$ are used to derive notions of
``behavioural conformance'' 
for $F$-coalgebras
(e.g. \cite{BaldanEA18,DBLP:journals/mscs/HasuoKC18,FGH+23,TBK+23}).

\subsection{Details for \autoref{sec:lax}}

\subsection*{Proof of \autoref{lem:symm}}
Immediate from involutivity of converse and the definition of
symmetry. \qed

\subsection*{Proof of \autoref{lem:lax-relconn}}
\begin{enumerate}
\item Immediate from the definition of $L\cdot L$.
\item `Only if' is clear. To see `if', use naturality; e.g.\
  $Ff=\Delta_{FY}\cdot Ff\subseteq L\Delta_Y\cdot Ff=L(\Delta_Y\cdot
  f)=Lf$ for $f\colon X\to Y$.
\item `Only if': Let $r\colon\frel{X}{Y}$. Then
  $\rev{L}r=\rev{(L\rev r)}=\rev{(\rev{(Lr)})}=Lr$. `If': Let
  $r\colon\frel{X}{Y}$. Then
  $L(\rev{r})=\rev{L}(\rev{r})=\rev{(L(\rev{(\rev{r})}))}=\rev{(Lr)}$.  \item This is clear by the previous items and the fact that
  condition~(L1) in the definition of lax extension (monotonicity)
  also features in the definition of relational connector.
\item Let $r\colon\frel{X}{Y}$. Then by hypothesis,
  $Lr\subseteq L\Delta_Y\cdot Lr\subseteq (L\cdot L)(\Delta_y\cdot r)=(L\cdot L)
  r$.
\item Immediate from~\ref{item:L2}. and~\ref{item:L3-comp}. \qed
\end{enumerate}

\subsection*{Details for \autoref{expl:Barr}}
Let~$F$ be a functor, and let~$L$ be a symmetric lax extension of~$F$, and write
$L_B$ for the Barr extension of~$F$. We show that $L_B\le L$. So let
$r\colon\frel{X}{Y}$, with projections $\pi_1,\pi_2$ as per
\autoref{rem:barr}. Then
\begin{align*}
  L_Br & = F\pi_1\cdot\rev{(F\pi_2)} \\
       & \le L\pi_1\cdot \rev{(L\pi_2)} & \by{(L3)}\\
       & = L\pi_1\cdot L(\rev{\pi_2}) & \by{symmetry}\\
       & \le L(\pi_1\cdot\rev{\pi_2}) & \by{(L2)}\\
       & = Lr.
\end{align*}

\subsection*{Proof of \autoref{thm:id-least}}
It is clear that~$\idcon{F}$ is transitive, and
by~\eqref{eq:idcon-diag},~$\idcon{F}$ extends~$F$, so by
\autoref{lem:lax-relconn}.\ref{item:lax},~$\idcon{F}$ is a lax
extension of~$F$. To see that~$\idcon{F}$ is symmetric, it suffices to
show that~$\rev{(\idcon{F})}$ is a right identity for composition of
relational connectors: For $L\colon F\to G$, we have
$L\cdot\rev{(\idcon{F})}=\rev{(\idcon{F}\cdot
  \rev{L})}=\rev{(\rev{L})}=L$ (using \autoref{lem:rev-comp}).
Finally, if~$F$ has some diagonal-preserving lax extension~$L$, then
$\idcon{F}\Delta_X\subseteq L\Delta_X=\Delta_{FX}$, i.e.~$\idcon{F}$
preserves diagonals. \qed

\subsection{Details for \autoref{sec:het-sim}}

\subsection*{Proof of \autoref{lem:sim-mor}}

We show first that $r\cdot f$ is an $L$-simulation. So let
$f(x)\mathrel{r} y$; we have to show
$\gamma'(x)\mathrel{L(r\cdot f)}\delta(y)$. By naturality, this is
equivalent to $Ff(\gamma'(x))\mathrel{Lr}\delta(y)$. Since~$f$ is an
$F$-coalgebra morphism, we have $Ff(\gamma'(x))=\gamma(f(x))$, and
$\gamma(f(x))\mathrel{Lr}\delta(y)$ because~$r$ is an $L$-simulation
and $f(x)\mathrel{r}y$.

Second, we show that $r\cdot\rev{g}$ is an $L$-simulation. Let
$x\mathrel{r}y$, so that $g(x)\mathrel{(r\cdot\rev{g})}y$. We have to
show that $\gamma''(g(x))\mathrel{L(r\cdot\rev
  g)}\delta(y)$. Since~$g$ is a $G$-coalgebra morphism, we have
$\gamma''(g(x))=Gg(\gamma(x))$, so the claim is, by naturality,
equivalent to $\gamma(x)\mathrel{L(r\cdot\rev g\cdot g)}\delta(y)$. By
monotonicity of~$L$ and totality of~$g$, it suffices to show
$\gamma(x)\mathrel{r}\delta(y)$, which follows from $x\mathrel{r}y$
because~$r$ is an $L$-simulation. \qed

\subsection*{Proof of \autoref{lem:functorial-simulation}} Let
$x\in C$, $z\in E$ such that $x\mathrel{(s\cdot r)} z$; i.e.\ we have
$y\in D$ such that $x\mathrel{r}y\mathrel{s}z$. Then
$\gamma(x)\mathrel{Kr}\delta(y)\mathrel{Ls}\varepsilon(z)$, so
$\gamma(x)\mathrel{(K\cdot L)(s\cdot r)}\varepsilon(z)$, as
required. \qed

\subsection*{Proof of \autoref{lem:sim-conv}}

We have
$\rev{r}\subseteq\rev{\gamma}\cdot
\rev{(Lr)}\cdot\delta=\rev{\gamma}\cdot \rev{L}(\rev
r)\cdot\delta$. \qed

\subsection*{Proof of \autoref{lem:sim-props}}
Throughout, we use
\begin{lemma}\label{lem:sim-monot}
  Let~$L,L'\colon F\to G$ be relational connectors such that
  $L\le L'$. If~$r\colon\frel{C}{D}$ is an $L$-(bi)simulation between
  an $F$-coalgebra $(C,\gamma)$ and a $G$-coalgebra $(D,\delta)$,
  then~$r$ is also an~$L'$-(bi)simulation.
\end{lemma}
Then, essentially all claims are immediate from
\autoref{lem:functorial-simulation} and \autoref{lem:sim-conv};
specifically:

  \begin{enumerate}
  \item Immediate from \autoref{lem:functorial-simulation}.
  \item Immediate from
    \autoref{lem:sim-conv}. 
  \item On an $F$-coalgebra $(C,\gamma)$,
    $\Delta_C$ 
    is an $L$-simulation because $\Delta_{FC}\subseteq
    L\Delta_C$. 
  \item Immediate from \autoref{lem:functorial-simulation} and
    \autoref{lem:sim-conv}. \qed
  \end{enumerate}

\subsection*{Proof of \autoref{cor:bisim-pres}}

Immediate from \autoref{lem:functorial-simulation},
\autoref{lem:sim-conv}, and \autoref{lem:sim-monot}. \qed

\subsection*{Details for \autoref{expl:bisim-transfer}}

We need to show that $L_R\cdot \rev{L_R}\le\id_G$. So let~$r=s\cdot t$
for $t\colon\frel{X}{Y}$, $s\colon\frel{Y}{Z}$, and let $S\in GX$,
$T\in FY$, $U\in GZ$ such that
$S\mathrel{L_Rt} T\mathrel{\rev{L_R}s} U$. We have to show that
$S\mathrel{\id_G r}U$, i.e.\ that~$S$ and~$U$ are related under the
usual Egli-Milner extension of~$r$. For the forward direction, let
$(b,x)\in S$; we have to find $(b,z)\in U$ such that
$x\mathrel{r}z$. Since~$R$ is right total, there is $a\in\A$ such that
$a\mathrel{R}b$, so by hypothesis we first obtain $(a,y)\in T$ such
that $x\mathrel{t}y$, and then $(b,z)\in U$ such that
$y\mathrel{s}z$. Since $r=s\cdot t$, we have $x\mathrel{r}z$, as
required. The back direction is analogous. \qed

\subsection*{Details for \autoref{expl:ioco}}

We check that $L$ is indeed a relational connector. To this end,
it is convenient to form it as a product $L_I \times L_O$
using \autoref{lem:converse-meet-product}.
Let $F = F_I \times F_O$ with $F_I(X) = (I \partialto X)$ and $F_O = (O \partialtone X)$,
and $G = G_I \times F_O$ with $G_I(X) = (I \to X)$. 
Then $L_I \colon F_I \rightarrow G_I$ and $L_O \colon F_O \rightarrow F_O$ are given by
\begin{align*}
	\delta_I \mathrel{L_I r} \tau_I &\iff \forall i \in \dom(\delta_I). \; \delta_I(i)\mathrel{r} \tau_I(i),\quad \text{and}  \\
	\delta_O \mathrel{L_O r} \tau_O &\iff \forall o \in \dom(\tau_O). \; o \in \dom(\delta_O) \text{ and } \delta_O(o)\mathrel{r} \tau_O(o).
\end{align*}
It suffices to check that $L_I$ and $L_O$ are both relational connectors. We focus on naturality.

Given a relation $r \colon \frel{X}{Y}$ and maps 
$f \colon X' \to X$ and $g \colon Y' \rightarrow Y$, we have:
\begin{align*}
	\delta_I \mathrel{L_I (\rev{g} \cdot r \cdot f)} \tau_I 
	&\iff \forall i \in \dom(\delta_I). \; \delta_I(i) \mathrel{\rev{g} \cdot r \cdot f} \tau_I(i) \\
	&\iff \forall i \in \dom(\delta_I). \; f(\delta_I(i)) \mathrel{r} g(\tau_I(i)) \\
	&\iff \forall i \in \dom(\delta_I). \; (F_If)(\delta_I)(i) \mathrel{r} (G_I g)(\tau_I)(i) \\
	&\iff \forall i \in \dom((F_If)(\delta_I)). \; (F_If)(\delta_I)(i) \mathrel{r} (G_I g)(\tau_I)(i) \\
	&\iff (F_I f)(\delta_I) \mathrel{L_I r} (G_I g)(\tau_I) \\
	&\iff \delta_I \mathrel{\rev{(G_I g)} \cdot L_I r \cdot (F_I f)} \tau_I \,.
\end{align*}

Naturality for $L_O$ follows in a similar manner:
\begin{align*}
	& \delta_O \mathrel{L_O (\rev{g} \cdot r \cdot f)} \tau_O \\
	&\iff \forall o \in \dom(\tau_O). \; o \in \dom(\delta_O) \text{ and } \delta_O(o) \mathrel{\rev{g} \cdot r \cdot f} \tau_O(o) \\
	&\iff \forall o \in \dom(\tau_O). \; o \in \dom(\delta_O) \text{ and } f(\delta_O(o)) \mathrel{r} g(\tau_O(o)) \\
	&\iff \forall o \in \dom(\tau_O). \; o \in \dom(\delta_O) \text{ and }  (F_Of)(\delta_O)(o) \mathrel{r} (F_O g)(\tau_O)(o) \\
	&\iff \forall o \in \dom((F_Og)(\tau_O)). \; o \in \dom((F_Of)(\delta_O)) \\
	& \qquad \qquad \text{ and } (F_Of)(\delta_O)(o) \mathrel{r} (F_O g)(\tau_O)(o) \\
	&\iff (F_O f)(\delta_O) \mathrel{L_O r} (F_O g)(\tau_O) \\
	&\iff \delta_O \mathrel{\rev{(F_O g)} \cdot L_O r \cdot (F_O f)} \tau_O \,.
\end{align*}

Next, we focus on the composition $\rev{L} \cdot L$, and prove that it indeed
corresponds to the presentation given in \autoref{expl:ioco}.
Let $r \colon \frel{X}{Z}$, and let $s \cdot t$ be its couniversal factorization,
with intermediate set $Y = \{(A,B) \mid A \times B \subseteq r\}$.
\autoref{thm:comp-couniv} tells us that $\rev{L} \cdot L r = \rev{L}s\cdot L t$.
We have:
\begin{align*}
	&(\delta_I, \delta_O) \mathrel{\rev{(L\rev{s})}\cdot L t} (\delta_I', \delta_O') \\
	&\iff \exists \tau_I, \tau_O.  \;
	\begin{array}{l}
	(\delta_I, \delta_O) \mathrel{Lt} (\tau_I, \tau_O), \quad \text{and} \\
	(\delta_I', \delta_O') \mathrel{L\rev{s}} (\tau_I, \tau_O)
	\end{array} \\
	& \iff \exists \tau_I, \tau_O. \;
	\left\{
    \begin{array}{l}
      \forall i \in \dom(\delta_I). \; \delta_I(i)\mathrel{t} \tau_I(i) \\
      \forall o \in \dom(\tau_O). \; o \in \dom(\delta_O) \text{ and } \delta_O(o)\mathrel{t} \tau_O(o) \\
      \forall i \in \dom(\delta_I'). \; \delta_I'(i)\mathrel{\rev{s}} \tau_I(i) \\
      \forall o \in \dom(\tau_O). \; o \in \dom(\delta_O') \text{ and } \delta_O'(o)\mathrel{\rev{s}} \tau_O(o) \\
    \end{array} \right.	\\
	& \iff \exists \tau_I, \tau_O. \;
	\left\{
    \begin{array}{l}
      \forall i \in \dom(\delta_I). \; \delta_I(i) \in \pi_1(\tau_I(i)) \\
      \forall o \in \dom(\tau_O). \; o \in \dom(\delta_O) \text{ and } \delta_O(o) \in \pi_1(\tau_O(o)) \\
      \forall i \in \dom(\delta_I'). \; \delta_I'(i) \in \pi_2(\tau_I(i)) \\
      \forall o \in \dom(\tau_O). \; o \in \dom(\delta_O') \text{ and } \delta_O'(o) \in \pi_2(\tau_O(o)) \\
    \end{array} \right.	
\end{align*}
We first claim that the existence of $\tau_I \colon I \to Y$ satisfying the first and third condition
is equivalent to $\forall i \in \dom(\delta_I) \cap \dom(\delta_I'). \; \delta_I(i) \mathrel{r} \delta_I'(i)$.
Indeed, given such a $\tau_I$ and $i \in \dom(\delta_I) \cap \dom(\delta_I')$, by first and third condition
there is a pair $(A,B) \in Y$ (so that $A \times B \subseteq r$) with $\delta_I(i) \in A$ and $\delta_I'(i) \in B$, hence 
$\delta_I(i) \mathrel{r} \delta_I'(i)$. 

Conversely, suppose that $\forall i \in \dom(\delta_I) \cap \dom(\delta_I'). \; \delta_I(i) \mathrel{r} \delta_I'(i)$.
Define $\tau_I$ by 
$$\tau_I(i) = 
\begin{cases} 
	(\{\delta_I(i)\}, \{\delta_I'(i)\}) & \text{ if } i \in \dom(\delta_I) \cap \dom(\delta_I') \\
	(\{\delta_I(i)\}, \emptyset) &\text{ if } i \in \dom(\delta_I) \setminus \dom(\delta_I') \\
	(\emptyset, \{\delta_I'(i)\}) &\text{ if } i \in \dom(\delta_I') \setminus \dom(\delta_I) \\
	(\emptyset, \emptyset) & \text{ otherwise}
\end{cases}\,.
$$
The first case is well-defined by assumption, and the necessary conditions are satisfied.

Next, we prove that the existence of $\tau_O \colon O \partialtone Y$ satisfying the second and fourth condition above is equivalent
to the statement $\exists o \in \dom(\delta_O) \cap \dom(\delta_O'). \; \delta_O(o) \mathrel{r} \delta_O'(o)$.
From left to right, from the type of $\tau_O$ there exists $o \in \dom(\tau_O)$,
and by assumption this means $o \in \dom(\delta_O) \cap \dom(\delta_O')$, $\delta_O(o) \in \pi_1(\tau_O(o))$
and $\delta_O'(o) \in \pi_2(\tau_O(o))$. Thus there is a pair $(A,B) \in Y$ with $\delta_O(o) \in A$ and $\delta_O'(o) \in B$.
Since $A \times B \subseteq r$, we have
$\delta_O(o) \mathrel{r} \delta_O'(o)$.

For the converse, 
we assume there exists $o \in \dom(\delta_O) \cap \dom(\delta_O')$ such that $\delta_O(o) \mathrel{r} \delta_O'(o)$, and
define $\tau_O$ by 
$$
\tau_O(o') = 
\begin{cases}
	(\{\delta_O(o)\}, \{\delta_O'(o)\}) & \text{ if } o'=o \\
	(\emptyset, \emptyset) & \text{ otherwise } 
\end{cases} \,.
$$
Then $\dom(\tau_O)$ is indeed non-empty and well-defined, and the necessary conditions are satisfied.

\subsection{Details for \autoref{sec:kantorovich}}

\subsection*{Proof of \autoref{thm:kantorovich}}

We check the conditions of \autoref{def:relational-connector}.

\emph{Monotonicity}: Immediate from the definition of~$L_\Lambda$ and
monotonicity of the modalities.

\emph{Naturality}: Let $r\colon\frel{X}{Y}$, $f\colon X'\to X$,
$g\colon Y'\to Y$, $a\in FX'$, $b\in GY'$. We split the claimed
equivalence
$Ff(a)\mathrel{L_\Lambda r}Gg(b)\iff a\mathrel{L_\Lambda(\rev{g}\cdot
  r\cdot f)} b$ into two implications:

`$\Rightarrow$': Let $Ff(a)\mathrel{L_\Lambda r}Gg(b)$,
$(\lambda,\mu)\in\Lambda$, and $A\subseteq X'$ such that
$a\in\lambda(A)$. We have to show that
$b\in\mu(\rev{g}\cdot r\cdot f[A])$, equivalently (by naturality of
predicate liftings) that $Gg(b)\in\mu(r\cdot f[A])$. By hypothesis,
this follows once we show that $Ff(a)\in\lambda(f[A])$, equivalently
$a\in\lambda(\rev{f}\cdot f[A])$; the latter follows from
$a\in\lambda(A)$ by monotonicity of~$\lambda$ and totality of~$f$.

`$\Leftarrow$': Let $a\mathrel{L_{\Lambda}(\rev{g}\cdot r\cdot f)} b$,
$(\lambda,\mu)\in\Lambda$, and $A\subseteq X$ such that
$Ff(a)\in\lambda(A)$; we have to show that $Gg(b)\in\mu(r[A])$,
equivalently $b\in\mu(\rev{g}\cdot r[A])$. Now we have
$a\in\lambda(\rev{f}[A])$ by naturality of~$\lambda$, and hence
$b\in\mu(\rev{g}\cdot r\cdot f\cdot\rev{f}[A])$ by hypothesis. But
$f\cdot\rev{f}[A]\subseteq A$, so by monotonicity of~$\mu$, we obtain
$b\in\mu(\rev{g}\cdot r[A])$ as required. \qed

\subsection*{Details for \autoref{expl:kantorovich}}

For \autoref{item:ioco-kant}, given $(\delta_I, \delta_O) \in FX$, 
$(\tau_I, \tau_O) \in GY$ and $r \colon \frel{X}{Y}$ we have 
\begin{align*}
	& (\delta_I, \delta_O) \mathrel{L_\Lambda} (\tau_I, \tau_O) \\
	& \iff \forall A \in 2^X. \;
		\left\{
		\begin{array}{l}
			\forall i \in I. \; \delta_I \in \Diamond_i(A) \Rightarrow \tau_I \in \Diamond_i(r[A]), \; \text{and}\\
			\forall o \in O. \; \delta_O \in \Box_o(A) \Rightarrow \tau_O \in \Box_o(r[A]), \; \text{and} \\
			\forall o \in O. \; \delta_O \in {\downarrow_o} \Rightarrow \tau_O \in {\downarrow_o} .
		\end{array}
		\right. 
\end{align*}
The first line above is equivalent to the statement $\forall i \in \dom(\delta_I). \; \delta_I(i) \mathrel{r} \tau_I(i)$.
To see this, from right to left, suppose $\forall i \in \dom(\delta_I). \; \delta_I(i) \mathrel{r} \tau_I(i)$, and let
$A \in 2^X$ and $i \in I$ such that $\delta_I \in \Diamond_i(A)$. This means that $i \in \dom(\delta_I)$ and
$\delta_I(i) \in A$. By assumption, from $i \in \dom(\delta_I)$ we get $\delta_I(i) \mathrel{r} \tau_I(i)$, and since
$\delta_I(i) \in A$ we get $\tau_I(i) \in r[A]$. Hence, $\tau_I \in \Diamond_i(r[A])$ (note that $i \in \dom(\tau_I)$ holds
since $\tau_I$ is total).

Conversely, assume the first line for all $A \in 2^X$ and suppose that 
$i\in \dom(\delta_I)$. Take $A = \{\delta_I(i)\}$; then $\delta_I \in \Diamond_i(A)$, hence $\tau_I \in \Diamond_i(r[A])$,
so that $\tau_I(i) \in r[A]$. Since $A = \{\delta_I(i)\}$, the latter implies $\delta_I(i) \mathrel{r} \tau_I(i)$.

The second and third line in the characterisation of the Kantorovich connector above, spelled out, say: 
\begin{align*}
\forall A \in 2^X, o \in O.
\begin{array}{l}
(o \in \dom(\delta_O) \Rightarrow \delta_O(o) \in A) \Rightarrow 
(o \in \dom(\tau_O) \Rightarrow \tau_O(o) \in r[A]); \\
(o \not \in \dom(\delta_O) \Rightarrow o \not \in \dom(\tau_O)).
\end{array}
\end{align*}
This is equivalent to:
\begin{align*}
	\forall A \in 2^X. \; \forall o \in \dom(\tau_O). \; o \in \dom(\delta_O) \text{ and } (\delta_O(o) \in A \Rightarrow \tau_O(o) \in r[A])
\end{align*}
This, in turn, is equivalent to 
$\forall o \in \dom(\tau_O). \; o \in \dom(\delta_O) \text{ and } \delta_O(o) \mathrel{r} \tau_O(o)$
as needed (again taking $A = \{\delta_O(o)\}$ in one direction).

\subsection*{Proof of \autoref{thm:kant-functorial}}
  \begin{enumerate}
  \item Let $t\colon\frel{X}{Y}$, $s\colon\frel{Y}{Z}$, $a\in FX$
    $b\in GY$, $c\in HZ$ such that
    $a\mathrel{L_\Lambda t} b\mathrel{L_\Theta s} c$; we have to show
    that $a\mathrel{L_{\Theta\cdot\Lambda}(s\cdot t)} c$. So let
    $(\lambda,\pi)\in\Theta\cdot\Lambda$, i.e.\ we have
    $\mu\in\predlift{G}$ such that $(\lambda,\mu)\in\Lambda$ and
    $(\mu,\pi)\in\Theta$; let $A\subseteq X$; and let
    $a\in\lambda(A)$.  We have to show that $c\in\pi(s\cdot
    t[A])$. But from $a\in\lambda(A)$ we obtain $b\in\mu(t[A])$
    because $a\mathrel{L_\Lambda t} b$, whence $c\in\pi(s\cdot t[A])$
    because $b\mathrel{L_\Theta s} c$.
  \item We show $\le$; the converse inequality then follows:
    \begin{equation*}
      L_{\rev{(\dual\Lambda)}}=\rev{(\rev{L_{\rev{(\dual\Lambda)}}})}\le\rev{L_{\rev{\big(\dual{\rev{(\dual\Lambda)}}\big)}}}=\rev{L_{\Lambda}}.
    \end{equation*}
    So let $r\colon \frel{X}{Y}$, and let $a\in GX$, $b\in FY$ such
    that $a\mathrel{\rev{L_\Lambda}r} b$,
    i.e. $b\mathrel{(L_\Lambda(\rev{r}))}a$. We have to show that
    $a\mathrel{L_{\rev{\dual\Lambda}}r} b$. So let
    $(\lambda,\mu)\in\Lambda$ and $A\subseteq X$ such that
    $a\in\dual\mu(A)$; we have to show that
    $b\in\dual\lambda(r[A])=FY\setminus\lambda(Y\setminus r[A])$. So
    assume that $b\in\lambda(Y\setminus r[A])$; then
    $a\in\mu(\rev{r}[Y\setminus r[A]])$ because
    $b\mathrel{(L_\Lambda(\rev{r}))}a$. By monotonicity of~$\mu$, this
    is in contradiction with
    $a\in\dual\mu(A)=GX\setminus\mu(X\setminus A)$ because
    $\rev{r}[Y\setminus r[A]]\subseteq X\setminus A$ (to see the
    latter, note that for~$x\in\rev{r}[Y\setminus r[A]]$ we have
    $x\mathrel{r}y$ for some $y\in Y\setminus r[A]$, so $x\notin
    A$). \qed
  \end{enumerate}

\subsection*{Proof of \autoref{cor:lax-kant}}

Claims~\ref{item:kant-trans} and~\ref{item:kant-symm} are immediate
from \autoref{thm:kant-functorial} (for Claim~\ref{item:kant-symm},
notice that the assumptions imply $\rev{(\dual\Lambda)}=\Lambda$). For
Claim~\ref{item:kant-lax}, transitivity of $L_\Lambda$ for
$\Lambda\subseteq\id$ is immediate from Claim~\ref{item:kant-trans};
moreover, for $\Lambda\subseteq\id$ it is trivial to note
that~$L_\Lambda$ extends~$F$, i.e.\ that
$\Delta_{FX}\subseteq L_\Lambda\Delta_X$ for
all~$X$. Claim~\ref{item:kant-lax-normal} is similarly immediate. \qed

\subsection*{Details for \autoref{rem:comp-kant}}

The proof of the inequality
$L_\Theta\cdot L_\Lambda\le L_{\Theta\logcomp\Lambda}$ is completely
analogous to that of
$L_\Theta\cdot L_\Lambda\le L_{\Theta\cdot\Lambda}$
(\autoref{thm:kant-functorial}.\ref{item:kant-comp}). To define
$\Lambda^\posbool$, we first define the set $\Pos(Z)$ of
\emph{positive Boolean combinations}~$\phi,\psi$ over a set~$Z$ by
\begin{equation*}
  \phi,\psi::=z\mid\bot\mid\top\mid\phi\lor\psi\mid\phi\land\psi\qquad(z\in Z),
\end{equation*}
and let~$V$ denote the set of placeholders $(-)_n$ for
$n\in\Nat$. Again for any set~$Z$, we write
$\Lambda(Z)=\{\doublemod{\lambda,\mu}(z_1,\dots,z_n)\mid(\lambda,\mu)\in\Lambda\text{
  $n$-ary}\}$. We then put
\begin{equation*}
  \Lambda^\posbool=\Pos(\Lambda(\Pos(V))).
\end{equation*}
For~$\phi\in\Lambda^\posbool$ mentioning placeholders
$(-)_1,\dots,(-)_k$ and subsets $A_1,\dots,A_k\subseteq X$, we
interpret $\phi(A_1,\dots,A_k)$ as a subset of~$FX$ recursively in the
obvious manner: We interpret~$(-)_i$ by~$A_i$, we interpret both inner
and outer occurrences of propositional operators as expected ($\land$
by intersection, $\lor$ by union etc.), and for
$(\lambda,\mu)\in\Lambda$, we interpret $\doublemod{\lambda,\mu}$ by
applying~$\lambda$ to the interpretations of its argument
formulae. Overall, we obtain an interpretation of~$\phi$ as a
predicate lifting for~$F$. Analogously, we have an interpretation
of~$\phi$ as predicate lifting for~$G$, so that~$\phi$ represents a
pair of $k$-ary predicate liftings.

In the claimed equality $L_\Lambda=L_{\Lambda^\posbool}$, `$\ge$' is
trivial because $\Lambda$ is contained in~$\Lambda^\posbool$ modulo a
minor shift in syntax caused by the use of placeholders. To show
`$\le$', let $r\colon\frel{X}{Y}$, $a\in FX$, and $b\in GY$ such that
$a\mathrel{L_{\Lambda}}b$. To show that
$a\mathrel{L_{\Lambda^\posbool}}b$, we have to show for
$\phi\in\Lambda^\posbool$ $k$-ary that whenever
$a\in\phi(A_1,\dots,A_k)$ for $A_1,\dots,A_k\in 2^X$, then
$b\in\phi(r[A_1],\dots,r[A_k])$ (where we interpret~$\phi$ as a
predicate lifting for~$F$ in the first instance and as a predicate
lifting for~$G$ in the second instance). This is by straightforward
induction on the structure of~$\phi$.

\subsection*{Details for \autoref{expl:kant-comp}}

Let $r\colon\frel{X}{Y}$, $S\in\pow(\A\times X)$,
$T\in\pow(\A\times Y)$ such that
$S
\mathrel{L_{\Lambda^{\posbool}\logcomp\rev{(\dual\Lambda^{\,\posbool})}}
  r} T$. We have to show that
$S\mathrel{(\tracecon\cdot\rev{{\tracecon}})r}T$. We factorize~$r$ as
$r= r\cdot\Delta_Y$. For $l\in\A$, put $A_l=\{x\mid(a,x)\in S\}$. Then
$S\models\Land_{l\in\A}\Box_l A_l$, so by hypothesis,
$T\models\Lor_{l\in\A}\Diamond_lr[A_l]$; thus, there exist $x\in A_l$,
$y\in Y$, $l\in\A$ such that $x\mathrel{r}y$ and $(l,y)\in T$. Since
moreover $(l,x)\in S$ by the definition of~$A_l$, this shows that
$S\mathrel{(\tracecon\cdot\rev{{\tracecon}})r}T$ by the description of
$\tracecon\cdot\rev{{\tracecon}}$ given in \autoref{expl:trace-con}.

\subsection{Details for \autoref{sec:expressiveness}}

\subsection*{Proof of \autoref{prop:adequacy}}

Let $\gamma\colon C\to FC$ be an $F$-coalgebra, let
$\delta\colon D\to GD$ be a $G$-coalgebra, and let
$r\colon\frel{C}{D}$ be an $L_\Lambda$-simulation from~$C$ to~$D$. We
proceed by induction on~$\phi$. Boolean cases are trivial. For the
modal case, let $x\in C$ such that
$x\models_F\doublemod{\lambda,\mu}\phi$, i.e.\
$\gamma(x)\in\lambda(\Sem{\phi}_C)$, and let $x\mathrel{r} y$. Then
$\gamma(x)\mathrel{L_\Lambda r}\delta(y)$, so
$\delta(y)\in\mu(r[\Sem{\phi}_C])$. By induction,
$r[\Sem{\phi}_C]\subseteq\Sem{\phi}_D$, so
$\delta(y)\in\mu(\Sem{\phi}_D)$ by monotonicity of~$\mu$, i.e.\
$y\models_G\doublemod{\lambda,\mu}\phi$. \qed

\subsection*{Proof of \autoref{thm:expr}}

Put
\begin{equation*}
  r=\{(x,y)\in C\times D\mid\forall
  \phi\in\FLang(\Lambda).\,x\models_F\phi\implies y\models_G\phi\}.
\end{equation*}
We show that~$r$ is an $L_\Lambda$-simulation. So let $x\mathrel{r}y$;
we have to show that $\gamma(x)\mathrel{L_\Lambda
  r}\delta(y)$. Since~$C$ and~$D$ are finitely branching, we have
finite subsets $C'\subseteq C$, $D'\subseteq D$ such that
$\gamma(x)\in FC'\subseteq FC$ and $\delta(y)\in GD'\subseteq GD$. Let
$A\subseteq C$ such that $\gamma(x)\in\lambda(A)$, equivalently
$\gamma(x)\in\lambda(C'\cap A)$; we have to show that
$\delta(y)\in\mu(r[A])$, equivalently $\delta(y)\in\mu(D'\cap
r[A])$. By monotonicity of~$\mu$, it suffices to show
$\delta(y)\in\mu(D'\cap r[A\cap C'])$. Assume the contrary. Put
$\FA=\{B\subseteq D'\mid \delta(y)\in\mu(B)\}$. Again by monotonicity
of~$\mu$, the assumption means that we have
$B\not\subseteq r[A\cap C']$ for every $B\in\FA$; we can thus pick
$y_B\in B\setminus r[A\cap C']$. Then $(x',y_B)\notin r$ for every
$x'\in C'\cap A$; that is, we can pick a $\Lambda$-formula
$\phi_{x',B}$ such that $x'\models_F\phi_{x',B}$ but
$y_B\not\models_G\phi_{x',B}$. Now put
\begin{equation*}\textstyle
  \psi=\Lor_{x'\in C'\cap A}\Land_{B\in\FA}\phi_{x',B}
\end{equation*}
(a finite formula because $C'$ and $\FA$ are finite).  Then
$ C'\cap A\subseteq\Sem{\psi}_C$ by construction, so
$x\models_F\doublemod{\lambda,\mu}\psi$ by monotonicity
of~$\lambda$. Since $x\mathrel{r} y$, we obtain that
$y\models_G\doublemod{\lambda,\mu}\psi$, i.e.\
$\delta(y)\in\mu(\Sem{\psi}_D)$, and hence
$\delta(y)\in\mu(D'\cap\Sem{\psi}_D)$. Thus,
$B:=D'\cap\Sem{\psi}_D\in\FA$, so we have $x'\in C'\cap A$ such that
$y_B\models\phi_{x',B}$, contradiction. \qed

\end{document}

%% file: defslncs.tex
\spnewtheorem{thm}[theorem]{Theorem}{\bfseries}{\itshape}
\spnewtheorem{cor}[theorem]{Corollary}{\bfseries}{\itshape}
\spnewtheorem{cnj}[theorem]{Conjecture}{\bfseries}{\itshape}
\spnewtheorem{lem}[theorem]{Lemma}{\bfseries}{\itshape}
\spnewtheorem{lemdefn}[theorem]{Lemma and Definition}{\bfseries}{\itshape}
\spnewtheorem{prop}[theorem]{Proposition}{\bfseries}{\itshape}
\spnewtheorem{defn}[theorem]{Definition}{\bfseries}{\upshape}
\spnewtheorem{rem}[theorem]{Remark}{\bfseries}{\upshape}
\spnewtheorem{notation}[theorem]{Notation}{\bfseries}{\upshape}
\spnewtheorem{expl}[theorem]{Example}{\bfseries}{\upshape}
\spnewtheorem{thmdefn}[theorem]{Theorem and Definition}{\bfseries}{\itshape}
\spnewtheorem{propdefn}[theorem]{Proposition and Definition}{\bfseries}{\itshape}
\spnewtheorem{assn}[theorem]{Assumption}{\bfseries}{\upshape}

 \renewenvironment{theorem}{\begin{thm}}{\end{thm}}
 \renewenvironment{corollary}{\begin{cor}}{\end{cor}}
 \renewenvironment{lemma}{\begin{lem}}{\end{lem}}
 \renewenvironment{proposition}{\begin{prop}}{\end{prop}}
 
 \renewenvironment{remark}{\begin{rem}}{\end{rem}}
 \renewenvironment{example}{\begin{expl}}{\end{expl}}
